\documentclass{CSML}
\pdfoutput=1

\usepackage{lastpage}
\newcommand{\dOi}{13(3:10)2017}
\lmcsheading{}{1--\pageref{LastPage}}{}{}%
{Nov.~30, 2016}{Aug.~15, 2017}{}

\usepackage{mathrsfs}
\usepackage{amssymb}
\usepackage{stmaryrd}
\usepackage{bussproofs}
\usepackage{mathtools}
\usepackage{bm}
\usepackage{thm-restate}
\usepackage{hyperref}
\hypersetup{
  hidelinks
}

\newcommand{\red}{\rightarrowtriangle}
\newcommand{\rtred}{\red^\ast}

\newcommand{\Suc}{\mathsf{S}}

\newcommand{\boolrec}{{\mathsf{rec}_\tbool}}
\newcommand{\unitrec}{{\mathsf{rec}_\tunit}}
\newcommand{\botrec}{{\mathsf{rec}_\tempt}}
\newcommand{\fpr}{.1}
\newcommand{\spr}{.2}
\newcommand{\natrec}{{\mathsf{rec}_\tnat}}
\newcommand{\MLTT}{\mathrm{MLTT}}
\newcommand{\MP}{\mathrm{MP}}
\newcommand{\iszero}{\mathsf{IsZero}}

\newcommand{\f}{\mathsf{f}}
\newcommand{\w}{\mathsf{w}}
\newcommand{\mw}{\mathsf{mw}}

\newcommand{\tnat}{\mathsf{N}}
\newcommand{\tempt}{{\mathsf{N}_0}}
\newcommand{\tunit}{{\mathsf{N}_1}}
\newcommand{\tbool}{{\mathsf{N}_2}}
\newcommand{\tuniv}{\mathsf{U}}
\newcommand{\lfor}{\mathrm{F}}
\newcommand{\tMP}{\mathsf{MP}}
\newcommand{\zer}{\mathsf{0}}
\newcommand{\one}{\mathsf{1}}

\newcommand{\nats}{\mathbb{N}}

\newcommand{\covt}{\vartriangleleft}
\newcommand{\dom}[1]{\textup{dom(}#1\textup{)}}

\newcommand{\forces}{\Vdash}
\newcommand{\nforces}{\nVdash}
\newcommand{\typ}{\!:\!}

\makeatletter
\newsavebox\myboxA
\newsavebox\myboxB
\newlength\mylenA

\newcommand*\xoverline[2][0.75]{%
    \sbox{\myboxA}{$\m@th#2$}%
    \setbox\myboxB\null
    \ht\myboxB=\ht\myboxA%
    \dp\myboxB=\dp\myboxA%
    \wd\myboxB=#1\wd\myboxA
    \sbox\myboxB{$\m@th\overline{\copy\myboxB}$}
    \setlength\mylenA{\the\wd\myboxA}
    \addtolength\mylenA{-\the\wd\myboxB}%
    \ifdim\wd\myboxB<\wd\myboxA%
       \rlap{\hskip 0.5\mylenA\usebox\myboxB}{\usebox\myboxA}%
    \else
        \hskip -0.5\mylenA\rlap{\usebox\myboxA}{\hskip 0.5\mylenA\usebox\myboxB}%
    \fi}
\makeatother

\begin{document}

\title[The Independence of Markov's Principle in Type Theory]{The Independence of Markov's Principle in Type Theory}

\author{Thierry Coquand}	
\address{University of Gothenburg, Sweden}	
\email{thierry.coquand@cse.gu.se}  

\author{Bassel Mannaa}	
\address{IT University of Copenhagen, Denmark}	
\email{basm@itu.dk}  



\keywords{Forcing, Dependent type theory, Markov’s Principle}
\subjclass{F.4.1 Mathematical Logic}
\titlecomment{This is an extended version of the paper by the same name published in the LIPIcs proceedings of FSCD'16}


\begin{abstract}
In this paper, we show that Markov's principle is not derivable in dependent type theory
with natural numbers and one universe. One way to prove this would be to remark
that Markov's principle does not hold in a sheaf model of type theory over Cantor space, since
Markov's principle does not hold for the generic point of this model \cite{CoquandMR17}. Instead we design an extension
of type theory, which intuitively extends type theory by the addition of a generic point of Cantor space. We
then show the consistency of this extension by a normalization argument. Markov's principle
does not hold in this extension, and it follows that it cannot be proved in type theory.
\end{abstract}
\maketitle

\section*{Introduction}
Markov's principle has a special status in constructive mathematics. One way to formulate this
principle is that if it is impossible that a given algorithm does not terminate, then it does
terminate. It is equivalent to the fact (Post's theorem) that if a set of natural number and its complement are
both computably enumerable, then this set is decidable \cite[Ch4]{vandalen}. This form is often used in recursion
theory. This principle was first formulated by Markov, who called it ``Leningrad's principle'', 
and founded a branch of constructive mathematics around this principle \cite{Margenstern1995}.

This principle is also equivalent to the fact that if a given real number is {\em not} equal to $0$
then this number is {\em apart} from $0$ (that is this number is $<-r$ or $>r$ for some 
rational number $r>0$). On this form, it was explicitly {\em refuted} by Brouwer in intuitionistic
mathematics, who gave an example of a real number (well defined intuitionistically) which is not
equal to $0$, but also not apart from $0$. (The motivation of Brouwer for this example was to show
the necessity of using {\em negation} in intuitionistic mathematics \cite{brouwer_negative}.) 
The idea of Brouwer can be represented formally using topological models \cite{vanDalen19781}.

In a neutral approach to mathematics, such as Bishop's \cite{Bishop67foundationsof}, Markov's principle is simply
left undecided. We also expect to be able to prove that Markov's principle is {\em not} provable
in formal system in which we can express Bishop's mathematics. For instance, Kreisel \cite{Kreisel1959-KREIOA} introduced {\em modified realizability} to show that Markov's principle is not derivable in the formal
system $HA^{\omega}$. Similarly, one would expect that Markov's principle is {\em not} derivable
in Martin-L\"of type theory \cite{martinloftypes}, but, as far as we know, such a result has not been 
established yet. \footnote{The paper \cite{Hyland93modifiedrealizability} presents a model of the calculus of constructions using the idea of modified
realizability, and it seems possible to use also this technique to interpret the type theory
we consider and prove in this way the independence of Markov's principle.}

We say that a statement $A$ is \emph{independent} of some formal system if $A$ cannot be derived in that system. A statement in the formal system of Martin-L\"of type theory ($\MLTT$) is represented by a closed type. A statement/type $A$ is derivable if it is inhabited by some term $t$ (written $\MLTT\vdash t\typ A$). This is the so-called propositions-as-types principle. Correspondingly we say that a statement $A$ (represented as a type) is independent of $\MLTT$ if there is no term $t$ such that $\MLTT\vdash t\typ A$.

The main result of this paper is to show that Markov's principle is independent
of Martin-L\"of type theory.\footnote{Some authors define independence in the stronger sense ``A statement is independent of a formal system if neither the statement nor its negation is provable in the system'', e.g.\ \cite{kunen}. We will also establish the independence of Markov's principle in this stronger sense.}

The main idea for proving this independence is to follow Brouwer's argument. We want to extend
type theory with a ``generic'' infinite sequence of $0$ and $1$ and establish that it is both absurd that this
generic sequence is never $0$, but also that we cannot show that it {\em has to} take the value $0$.
To add such a generic sequence is exactly like adding a {\em Cohen real} \cite{cohen} in forcing
extension of set theory. A natural attempt for doing this will be to consider a {\em topological
model} of type theory (sheaf model over Cantor space), extending the work \cite{vanDalen19781} to type
theory. However, while it is well understood how to represent universes in {\em presheaf} model
\cite{liftinggrothendieck}, it has turned out to be surprisingly difficult to represent universes
in {\em sheaf} models, see \cite{sheafuniverse}, \cite{streicher_Universe_in_toposes}. Also see \cite{CoquandMR17} for a possible solution. Our approach is here instead a purely {\em syntactical} description of a forcing extension of type theory (refining previous work of \cite{coquand2010note}),
which contains a formal symbol for the generic sequence and a proof that it is absurd that this generic
sequence is never $0$, together with a {\em normalization} theorem, from which we can deduce that we
{\em cannot} prove that this generic sequence has to take the value $0$. Since this formal system is
an extension of type theory, the independence of Markov's principle follows.

As stated in \cite{Kopylov2001}, which describes an elegant generalization of this principle in type theory,
Markov's principle is an important technical tool for proving termination of computations, and thus
can play a crucial role if type theory is extended with general recursion as in \cite{constablepartial}.

This paper is organized as follows. We first describe the rules of the version of type theory
we are considering. This version can be seen as a simplified version of type theory as represented
in the system Agda \cite{norell:thesis}, and in particular, contrary to the work \cite{coquand2010note}, we allow
$\eta$-conversion, and we express conversion as {\em judgment}. 
Markov's principle can be formulated in a natural way in this formal system.
We describe then the forcing extension of type theory, where we add a Cohen real. For proving
normalization, we follow Tait's computability method \cite{Tait1967-TAIIIO, martinloftypes}, but we have to consider an extension of this with a computability {\em relation} in order to interpret the conversion judgment.
This can be seen as a forcing extension of the technique used in \cite{abelScherer:types10}. Using this computability argument, it is then possible to show that we cannot show that the generic sequence has to take the value $0$.
We end by a refinement of this method, giving a consistent extension of type theory where the
{\em negation} of Markov's principle is provable.

\section{Type theory and forcing extension}
\label{SEC:Forcingextensionoftypetheory}

The syntax of our type theory is given by the grammar:
\begin{align*}
 t,u,A,B  \coloneqq &\, x \mid \botrec (\lambda x. A) \mid \unitrec (\lambda x. A)\,t 
                  \mid \boolrec(\lambda x. A)\,t\,u \mid \natrec (\lambda x. A)\,t\,u  \\
                  & \mid \tuniv \mid \tnat\mid  \tempt \mid \tunit\mid \tbool \mid \zer \mid \one \mid \Suc\,t
                  \mid \Pi (x\typ A) B \mid \lambda x. t \mid t\,u \mid \Sigma (x\typ A)  B\mid (t, u) \mid t\fpr \mid t\spr 
\end{align*}

We use the notation $\xoverline{n}$ as a short hand for the term $\Suc^n\,\zer$, where $\Suc$ is the successor constructor. We will use the symbol $\coloneqq$ for definitional equality in the metatheory.

\subsection{Type system}
We describe a type theory with one universe \`{a} la Russell, natural numbers, functional extensionality and surjective pairing, hereafter referred to as $\MLTT$.\footnote{This is a type system similar to Martin-l{\"o}f's \cite{martinloftypes} except that we have $\eta$-conversion and surjective pairing.}\bigskip

 
\noindent
\textbf{Natural numbers:}
\begin{center} 	
\AxiomC{$\Gamma \vdash$ }
\UnaryInfC{$\Gamma \vdash \tnat$}
\DisplayProof
\,
\AxiomC{$\Gamma \vdash$}
\UnaryInfC{$\Gamma \vdash \zer \typ  \tnat$}
\DisplayProof
\,
\AxiomC{$\Gamma \vdash n \typ  \tnat$}
\UnaryInfC{$\Gamma \vdash \Suc\, n \typ  \tnat$}
\DisplayProof 
\,
\AxiomC{$\Gamma \vdash n = m\typ  \tnat$}
\UnaryInfC{$\Gamma \vdash \Suc\, n = \Suc\, m\typ  \tnat$}
\DisplayProof \medskip

\begin{prooftree}
\AxiomC{$\Gamma, x\typ \tnat \vdash F$}
\AxiomC{$\Gamma \vdash a_0 \typ  F[\zer]$}
\AxiomC{$\Gamma \vdash g\typ  \Pi (x\typ \tnat) (F[x]\rightarrow F[\Suc\,x])$}
\TrinaryInfC{$\Gamma \vdash \natrec(\lambda x. F)\,a_0\, g\typ  \Pi (x\typ \tnat) F$}
\end{prooftree}

\begin{prooftree}
\AxiomC{$\Gamma, x\typ \tnat \vdash F$}
\AxiomC{$\Gamma \vdash a_0 \typ  F[\zer]$}
\AxiomC{$\Gamma \vdash g \typ  \Pi (x\typ \tnat) (F[x] \rightarrow F[\Suc\,x])$}
\TrinaryInfC{$\Gamma \vdash \natrec (\lambda x. F)\,a_0\, g \,\zer = a_0 \typ  F[\zer]$}
\end{prooftree}

\begin{prooftree}
\AxiomC{$\Gamma, x\typ \tnat \vdash F$}
\AxiomC{$\Gamma \vdash a_0 \typ  F[\zer]$}
\AxiomC{$\Gamma \vdash n \typ  \tnat$}
\AxiomC{$\Gamma \vdash g \typ  \Pi (x\typ \tnat) (F[x] \rightarrow F[\Suc \,x])$}
\QuaternaryInfC{$\Gamma \vdash \natrec (\lambda x. F)\,a_0\, g \,(\Suc\, n) = g\,n \,(\natrec (\lambda x. F)\,a_0\,g\,n) \typ  F[\Suc\,n]$}
\end{prooftree}

\begin{prooftree}
\AxiomC{$\Gamma, x\typ \tnat\vdash F = G$}
\AxiomC{$\Gamma \vdash a_0=b_0 \typ  F[\zer]$}
\AxiomC{$\Gamma \vdash g=h \typ  \Pi (x\typ \tnat) (F[x] \rightarrow F[\Suc \,x])$}
\TrinaryInfC{$\Gamma \vdash \natrec (\lambda x. F)\,a_0\, g = \natrec (\lambda x. G)\,b_0\, h\typ  \Pi(x\typ \tnat) F$}
\end{prooftree}
\end{center}\medskip

\noindent
\textbf{Booleans:}

\begin{center}
\AxiomC{$\Gamma \vdash$ }
\UnaryInfC{$\Gamma \vdash \tbool$}
\DisplayProof
\,
\AxiomC{$\Gamma \vdash$}
\UnaryInfC{$\Gamma \vdash \zer \typ  \tbool$}
\DisplayProof
\,
\AxiomC{$\Gamma \vdash$}
\UnaryInfC{$\Gamma \vdash \one \typ  \tbool$}
\DisplayProof \,
\AxiomC{$\Gamma, x\typ \tbool \vdash F$}
\AxiomC{$\Gamma \vdash a_0 \typ  F[\zer]$}
\AxiomC{$\Gamma \vdash a_1 \typ  F[\one]$}
\TrinaryInfC{$\Gamma \vdash \boolrec(\lambda x. F)\,a_0\, a_1\typ  \Pi (x\typ \tbool) F$}
\DisplayProof \bigskip

\AxiomC{$\Gamma, x\typ \tbool \vdash F$}
\AxiomC{$\Gamma \vdash a_0 \typ  F[\zer]$}
\AxiomC{$\Gamma \vdash a_1 \typ  F[\one]$}
\TrinaryInfC{$\Gamma \vdash \boolrec (\lambda x. F)\,a_0\, a_1 \,\zer = a_0 \typ  F[\zer]$}
\DisplayProof
\,
\AxiomC{$\Gamma, x\typ \tbool \vdash F$}
\AxiomC{$\Gamma \vdash a_0 \typ  F[\zer]$}
\AxiomC{$\Gamma \vdash a_1 \typ  F[\one]$}
\TrinaryInfC{$\Gamma \vdash \boolrec (\lambda x. F)\,a_0\, a_1 \,\one = a_1 \typ  F[\one]$}
\DisplayProof \bigskip

\noindent
\AxiomC{$\Gamma, x\typ \tbool\vdash F = G$}
\AxiomC{$\Gamma \vdash a_0=b_0 \typ  F[\zer]$}
\AxiomC{$\Gamma \vdash a_1=b_1 \typ  F[\one]$}
\TrinaryInfC{$\Gamma \vdash \boolrec (\lambda x. F)\,a_0\, a_1 = \boolrec (\lambda x. G)\,b_0\, b_1\typ  \Pi(x\typ \tbool) F$}
\DisplayProof \bigskip
\end{center}\medskip

\noindent
\textbf{Unit Type:}

\begin{center}
\AxiomC{$\Gamma \vdash$ }
\UnaryInfC{$\Gamma \vdash \tunit$}
\DisplayProof
\,
\AxiomC{$\Gamma \vdash$}
\UnaryInfC{$\Gamma \vdash \zer \typ  \tunit$}
\DisplayProof
\,
\AxiomC{$\Gamma, x\typ \tunit \vdash F$}
\AxiomC{$\Gamma \vdash a \typ  F[\zer]$}
\BinaryInfC{$\Gamma \vdash \unitrec(\lambda x. F)\,a\typ  \Pi (x\typ \tunit) F$}
\DisplayProof \bigskip

\noindent
\AxiomC{$\Gamma, x\typ \tunit \vdash F$}
\AxiomC{$\Gamma \vdash a \typ  F[\zer]$}
\BinaryInfC{$\Gamma \vdash \unitrec (\lambda x. F)\,a \,\zer = a \typ  F[\zer]$}
\DisplayProof
\,
\AxiomC{$\Gamma, x\typ \tunit\vdash F = G$}
\AxiomC{$\Gamma \vdash a=b \typ  F[\zer]$}
\BinaryInfC{$\Gamma \vdash \unitrec (\lambda x. F)\,a = \unitrec (\lambda x. G)\,b\typ  \Pi(x\typ \tunit) F$}
\DisplayProof 
\end{center}\medskip

\noindent
\textbf{Empty type:}\smallskip

\begin{center}
\AxiomC{$\Gamma \vdash$ }
\UnaryInfC{$\Gamma \vdash \tempt$}
\DisplayProof
\,
\noindent
\AxiomC{$\Gamma, x\typ \tempt \vdash F$}
\UnaryInfC{$\Gamma \vdash \botrec(\lambda x. F)\typ  \Pi (x\typ \tempt) F$}
\DisplayProof
\,
\noindent
\AxiomC{$\Gamma, x\typ \tempt\vdash_p F = G$}
\UnaryInfC{$\Gamma \vdash  \botrec(\lambda x. F) =  \botrec(\lambda x. G)\typ  \Pi (x\typ \tempt) F$}
\DisplayProof
\end{center}
\medskip

\noindent
\textbf{Dependent functions:}\smallskip

\begin{center}
\AxiomC{$\Gamma \vdash F$ }
\AxiomC{$\Gamma, x\typ  F \vdash G$}
\BinaryInfC{$\Gamma \vdash \Pi (x\typ  F) G$}
\DisplayProof\,
\noindent
\AxiomC{$\Gamma\vdash F= H$}
\AxiomC{$\Gamma, x\typ  F \vdash G =E $}
\BinaryInfC{$\Gamma \vdash \Pi (x\typ F) G =\Pi (x\typ H) E$}
\DisplayProof\,
\AxiomC{$\Gamma, x \typ  F \vdash t \typ  G$ }
\UnaryInfC{$\Gamma \vdash \lambda x.t \typ  \Pi (x\typ F) G$}
\DisplayProof\bigskip

\noindent
\AxiomC{$\Gamma\vdash g \typ  \Pi (x\typ F) G$}
\AxiomC{$\Gamma\vdash a \typ  F$}
\BinaryInfC{$\Gamma \vdash g\,a \typ  G[a]$}
\DisplayProof\,
\AxiomC{$\Gamma \vdash g \typ  \Pi (x\typ F) G$}
\AxiomC{$\Gamma \vdash u =v \typ  F$}
\BinaryInfC{$\Gamma \vdash g\,u = g\,v \typ  G[u]$}
\DisplayProof
\,
\AxiomC{$\Gamma \vdash h =g \typ  \Pi (x\typ F) G$}
\AxiomC{$\Gamma \vdash u \typ  F$}
\BinaryInfC{$\Gamma \vdash h\,u = g\,u \typ  G[u]$}
\DisplayProof \bigskip

\AxiomC{$\Gamma, x\typ F \vdash t \typ  G$}
\AxiomC{$\Gamma \vdash a \typ  F$}
\BinaryInfC{$\Gamma \vdash (\lambda x. t) a = t[a] \typ  G[a]$}
\DisplayProof\,
\AxiomC{$\Gamma \vdash h \typ \Pi(x\typ  F)G$}
\AxiomC{$\Gamma \vdash g \typ \Pi(x\typ  F) G$}
\AxiomC{$\Gamma, x\typ F \vdash h\,x=g\,x \typ G[x]$}
\TrinaryInfC{$\Gamma \vdash h = g \typ  \Pi(x\typ  F)G$}
\DisplayProof \bigskip
\end{center}

\noindent
\textbf{Dependent pairs:}\smallskip
\begin{center}
\AxiomC{$\Gamma \vdash F$ }
\AxiomC{$\Gamma, x\typ F \vdash G $}
\BinaryInfC{$\Gamma \vdash \Sigma (x\typ F) G$}
\DisplayProof\,
\noindent
\AxiomC{$\Gamma\vdash F = H$}
\AxiomC{$\Gamma, x\typ  F \vdash G = E $}
\BinaryInfC{$\Gamma \vdash \Sigma(x\typ F)G =\Sigma(x\typ H)E$}
\DisplayProof
\,
\AxiomC{$\Gamma, x\typ F \vdash G$ }
\AxiomC{$\Gamma \vdash a \typ  F$ }
\AxiomC{$\Gamma \vdash b\typ  G[a]$}
\TrinaryInfC{$\Gamma \vdash (a, b) \typ  \Sigma (x\typ  F) G $}
\DisplayProof\bigskip

\AxiomC{$\Gamma\vdash t\typ  \Sigma (x\typ F) G$ }
\UnaryInfC{$\Gamma \vdash t\fpr \typ  F$ }
\DisplayProof
\,
\AxiomC{$\Gamma\vdash t\typ  \Sigma (x\typ F) G$ }
\UnaryInfC{$\Gamma \vdash t\spr \typ  G[t\fpr]$ }
\DisplayProof\,
\AxiomC{$\Gamma, x\typ F\vdash G$}
\AxiomC{$\Gamma\vdash t\typ F$}
\AxiomC{$\Gamma\vdash u\typ G[t]$}
\TrinaryInfC{$\Gamma \vdash (t,u)\fpr=t \typ  F$ }
\DisplayProof\bigskip

\AxiomC{$\Gamma, x\typ F\vdash G$}
\AxiomC{$\Gamma\vdash t\typ F$}
\AxiomC{$\Gamma\vdash u\typ G[t]$}
\TrinaryInfC{$\Gamma \vdash (t,u)\spr=u \typ  G[t]$ }
\DisplayProof\,
\AxiomC{$\Gamma \vdash t =u \typ  \Sigma (x\typ F) G$}
\UnaryInfC{$\Gamma \vdash t\fpr=u\fpr\typ F$}
\DisplayProof
\,
\AxiomC{$\Gamma \vdash t =u \typ  \Sigma (x\typ F) G$}
\UnaryInfC{$\Gamma \vdash t\spr=u\spr\typ G[t\fpr]$}
\DisplayProof \bigskip

\noindent
\AxiomC{$\Gamma\vdash t\typ  \Sigma(x\typ  F) G$}
\AxiomC{$\Gamma\vdash u\typ  \Sigma(x\typ  F) G$}
\AxiomC{$\Gamma\vdash t\fpr=u\fpr\typ  F$ }
\AxiomC{$\Gamma\vdash t\spr=u\spr\typ  G[t\fpr]$ }
\QuaternaryInfC{$\Gamma \vdash t=u \typ  \Sigma(x\typ F)G$ }
\DisplayProof
\medskip
\end{center}

\noindent
\textbf{Universe:}\smallskip

\begin{center}
\AxiomC{$\Gamma \vdash$ }
\UnaryInfC{$\Gamma \vdash \tuniv$}
\DisplayProof\,
\AxiomC{$\Gamma \vdash F \typ  \tuniv$}
\UnaryInfC{$\Gamma \vdash F$}
\DisplayProof\,
\AxiomC{$\Gamma \vdash F = G\typ  \tuniv$}
\UnaryInfC{$\Gamma \vdash F=G$}
\DisplayProof\,
\noindent
\AxiomC{$\Gamma \vdash$ }
\UnaryInfC{$\Gamma \vdash \tempt \typ \tuniv$}
\DisplayProof
\,
\AxiomC{$\Gamma \vdash$ }
\UnaryInfC{$\Gamma \vdash \tunit \typ \tuniv$}
\DisplayProof
\,
\AxiomC{$\Gamma \vdash$ }
\UnaryInfC{$\Gamma \vdash \tbool \typ \tuniv$}
\DisplayProof
\,
\AxiomC{$\Gamma \vdash$ }
\UnaryInfC{$\Gamma \vdash \tnat \typ \tuniv$}
\DisplayProof\bigskip

\noindent
\AxiomC{$\Gamma \vdash F \typ \tuniv$ }
\AxiomC{$\Gamma, x\typ  F \vdash G \typ \tuniv$}
\BinaryInfC{$\Gamma \vdash \Pi (x\typ  F) G \typ \tuniv$}
\DisplayProof\medskip
\,
\AxiomC{$\Gamma\vdash F= H \typ \tuniv$}
\AxiomC{$\Gamma, x\typ  F \vdash G =E \typ \tuniv$}
\BinaryInfC{$\Gamma \vdash \Pi (x\typ F) G =\Pi (x\typ H) E \typ \tuniv$}
\DisplayProof\bigskip

\noindent
\AxiomC{$\Gamma \vdash F \typ \tuniv$ }
\AxiomC{$\Gamma, x\typ  F \vdash G \typ \tuniv$}
\BinaryInfC{$\Gamma \vdash \Sigma (x\typ  F) G\typ  \tuniv$}
\DisplayProof\medskip
\,
\AxiomC{$\Gamma\vdash F= H \typ \tuniv$}
\AxiomC{$\Gamma, x\typ  F \vdash G =E \typ \tuniv$}
\BinaryInfC{$\Gamma \vdash \Sigma (x\typ F) G =\Sigma (x\typ H) E \typ \tuniv$}
\DisplayProof\medskip
\end{center}
\noindent
\textbf{Congruence:}

\begin{center}
\noindent
\AxiomC{$\Gamma \vdash F$}
\UnaryInfC{$\Gamma \vdash F=F$}
\DisplayProof\,
\AxiomC{$\Gamma \vdash F=G$}
\UnaryInfC{$\Gamma \vdash G=F$}
\DisplayProof\,
\AxiomC{$\Gamma \vdash F=G$}
\AxiomC{$\Gamma \vdash G=H$}
\BinaryInfC{$\Gamma \vdash F=H$}
\DisplayProof\,
\AxiomC{$\Gamma \vdash t\typ F$}
\UnaryInfC{$\Gamma \vdash t=t\typ F$}
\DisplayProof\,
\AxiomC{$\Gamma \vdash t=u\typ F$}
\UnaryInfC{$\Gamma \vdash u=t\typ F$}
\DisplayProof\bigskip

\AxiomC{$\Gamma \vdash t=u\typ F$}
\AxiomC{$\Gamma \vdash u=v \typ F$}
\BinaryInfC{$\Gamma \vdash t=v \typ F$}
\DisplayProof\,
\AxiomC{$\Gamma \vdash t\typ F$}
\AxiomC{$\Gamma \vdash F = G$}
\BinaryInfC{$\Gamma \vdash t\typ G$}
\DisplayProof\,
\AxiomC{$\Gamma \vdash t=u\typ F$}
\AxiomC{$\Gamma \vdash F = G$}
\BinaryInfC{$\Gamma \vdash t=u\typ G$}
\DisplayProof \bigskip

\AxiomC{$\Gamma \vdash a\typ A$}
\AxiomC{$\Gamma \vdash A=B$}
\BinaryInfC{$\Gamma \vdash a \typ B$}
\DisplayProof
\,
\AxiomC{$\Gamma \vdash a=b\typ A$}
\AxiomC{$\Gamma \vdash A=B$}
\BinaryInfC{$\Gamma \vdash a=b \typ B$}
\DisplayProof\bigskip

\end{center}

\normalsize

\noindent
The following four rules are admissible in the this type system \cite{abelScherer:types10}, we consider them as rules of our type system:\medskip

\noindent
\AxiomC{$\Gamma \vdash a\typ A$}
\UnaryInfC{$\Gamma\vdash A$}
\DisplayProof
\,
\AxiomC{$\Gamma\vdash a=b \typ A$}
\UnaryInfC{$\Gamma\vdash a \typ A$}
\DisplayProof
\,
\AxiomC{$\Gamma, x\typ F\vdash G$}
\AxiomC{$\Gamma\vdash a=b\typ F$}
\BinaryInfC{$\Gamma\vdash G[a]=G[b]$}
\DisplayProof
\,
\AxiomC{$\Gamma, x\typ F\vdash t \typ G$}
\AxiomC{$\Gamma\vdash a=b\typ F$}
\BinaryInfC{$\Gamma\vdash t[a]=t[b] \typ G[a]$}
\DisplayProof\medskip
\,
\AxiomC{$\Gamma \vdash A = B$}
\UnaryInfC{$\Gamma \vdash A$}
\DisplayProof

\subsection{Markov's principle}

Markov's principle can be represented in type theory by the type 
\begin{equation*}
\tMP\coloneqq \Pi(h\typ \tnat\rightarrow \tbool)[\neg\neg (\Sigma (x\typ \tnat)\,\iszero\,(h\,x)) \rightarrow \Sigma (x\typ \tnat)\,\iszero\,(h\,x)]
\end{equation*}
where $\iszero\typ \tbool \rightarrow \tuniv$ is defined by $\iszero \coloneqq \lambda y.\boolrec (\lambda x. \tuniv)\,\tunit\,\tempt\,y$. 

Note that $\iszero\,(h\,n)$ is inhabited when $h\,n = 0$ and empty when $h\,n = 1$. Thus $\Sigma(x\typ \tnat) \,\iszero\,(h\,x)$ is inhabited if there is $n$ such that $h\,n = 0$. 

We remark that in the presence of propositional truncation $\parallel . \parallel$, Markov's principle can be alternatively formulated with weak (propositional) existential $\exists (x\typ A) B \coloneqq \parallel \Sigma (x\typ A) B\parallel$. However, the two formulations are logically equivalent \cite[Exercise 3.19]{hottbook}.

The main result of this paper is the following:

\begin{restatable}{thm}{mainresult}
\label{THM:mainResult}
There is no term $t$ such that $\MLTT\vdash t\typ \tMP$.
\end{restatable}

An {\em extension} of $\MLTT$ is given by introducing new objects, judgment forms and derivation rules. This means in particular that any judgment valid in $\MLTT$ is valid in the extension.  A {\em consistent} extension is one in which the type $\tempt$ is uninhabited.

To show Theorem \ref{THM:mainResult} we will form a consistent extension of $\MLTT$ with a new constant $\vdash \f \typ \tnat\rightarrow \tbool$. We will then show that $\neg\neg (\Sigma (x\typ \tnat)\,\iszero\,(\f\,x))$ is derivable while $\Sigma (x\typ \tnat)\,\iszero\,(\f\,x)$ is not derivable. Thus showing that $\tMP$ is not derivable in this extension and consequently not derivable in $\MLTT$. 

While this is sufficient to establish independence in the sense of non-derivability of $\MP$, to establish the independence of $\MP$ in the stronger sense one also needs to show that $\neg \tMP$ is not derivable in $\MLTT$. This can achieved by reference to the work of Aczel \cite{Aczel1999} where it is shown that $\MLTT$ extended with $\vdash \mathsf{dne} \typ \Pi(A\typ \tuniv) (\neg \neg A \rightarrow A)$ is consistent. Since $h\typ \tnat\rightarrow \tbool, x\typ \tnat \vdash \iszero\,(h\,x) \typ \tuniv$ we have $h\typ \tnat\rightarrow \tbool\vdash \Sigma(x\typ \tnat)\,\iszero\,(h\,x)\typ \tuniv$. If we let $\mathsf{T}(h)\coloneqq \Sigma(x\typ \tnat)\,\iszero\,(h\,x)$ we get that $h\typ \tnat \rightarrow \tbool \vdash \mathsf{dne}\, \mathsf{T}(h) \typ \neg \neg \mathsf{T}(h) \rightarrow \mathsf{T}(h)$.
By $\lambda$ abstraction we have $\vdash \lambda h. \mathsf{dne}\,T(h) \typ \tMP$. We can then conclude that there is no term $t$ such that $\MLTT \vdash t\typ \neg \tMP$.

Finally, we will refine the result of Theorem \ref{THM:mainResult} by building a consistent extension of $\MLTT$ where $\neg \tMP$ is derivable.

\subsection{Forcing extension}

A \emph{condition} $p$ is a graph of a partial finite function from $\nats$ to $\{\zer,\one\}$. We denote by $\emptyset$ the empty condition. We write $p(n)=b$ when $(n,b) \in p$. We say $q$ \emph{extends} $p$ (written $q \leqslant p$) if $p$ is a subset of $q$. A condition can be thought of as a basic compact open in Cantor space $2^\nats$. Two conditions $p$ and $q$ are \emph{compatible} if $p\cup q$ is a condition and we write $p
q$ for $p\cup q$. If $n\notin \dom{p}$ we write $p(n\mapsto \zer)$ for $p\cup\{(n,\zer)\}$ and $p(n\mapsto \one)$ for $p\cup\{(n,\one)\}$. We define the notion of \emph{partition} corresponding to the notion of finite covering of a compact open in Cantor space.

\begin{defi}We write $p\covt S$ to say that $S$ is a partition of $p$ and we define it inductively as follows:
\begin{enumerate}
\item $p \covt \{p\}$.
\item If $n\notin \dom{p}$ and $p(n\mapsto \zer)\covt S_0$ and $p(n\mapsto \one)\covt S_1 $ then $p \covt S_0 \cup S_1$.
\end{enumerate}
\end{defi}

Note that if $p\covt S$ then any $q\in S$ and $r\in S$ are incompatible unless $q=r$. If moreover $s \leqslant p$ then $s \covt \{sq \mid q\in S \text{ compatible with } s\}$. 

We extend the given type theory by annotating the judgments with conditions, i.e.\ replacing each judgment $\Gamma \vdash J$ in the given type system with a judgment $\Gamma\vdash_p J$.\medskip

In addition, we add the locality rule: 
\AxiomC{$\Gamma \vdash_{p_1} J$}
\AxiomC{$\dots$}
\AxiomC{$\Gamma \vdash_{p_n} J $} 
\def\labelSpacing{1pt}
\LeftLabel{\scriptsize{\textsc{loc}}}
\RightLabel{$p \covt \{p_1,\dots, p_n\}$}
\TrinaryInfC{$\Gamma \vdash_p J $}
\DisplayProof\smallskip

We add a term $\f$ for the generic point along with the introduction and conversion rules:\smallskip
\begin{center}
\AxiomC{$\Gamma \vdash_p$}
\LeftLabel{\scriptsize{\textsc{$\f$-I}}}
\UnaryInfC{$\Gamma \vdash_p \f \typ  \tnat \rightarrow \tbool$}
\DisplayProof\,
\AxiomC{$\Gamma \vdash_p$}
\LeftLabel{\scriptsize{\textsc{$\f$-eval}}}
\RightLabel{$n \in \dom{p}$}
\UnaryInfC{$\Gamma \vdash_p \f\,\xoverline{n} = p(n)\typ \tbool$}
\DisplayProof \bigskip
\end{center}

We add a term $\w$ and the rule: 
\begin{center}
\AxiomC{$\Gamma \vdash_p$}
\LeftLabel{\scriptsize{\textsc{$\w$-term}}}
\UnaryInfC{$\Gamma \vdash_p \w\typ  \neg \neg (\Sigma(x\typ \tnat)\,\iszero\,(\f\,x)) $}
\DisplayProof
\end{center}

Since $\w$ inhabits the type $\neg \neg (\Sigma(x\typ \tnat)\,\iszero\,(\f\,x))$, our goal is then to show that no term inhabits the type $\Sigma (x\typ \tnat)\, \iszero(\f\,x)$.

It follows directly from the description of the forcing extension that:
\begin{lem}
\label{LEM:extension}
If $\Gamma\vdash J$ in standard type theory then $\Gamma \vdash_{\emptyset} J$.
\end{lem}
Note that if $q\leqslant p$ and $\Gamma\vdash_p J$ then $\Gamma\vdash_q J$ (monotonicity). A statement $A$ (represented as a closed type) is derivable in this extension if $\vdash_{\emptyset} t\typ A$ for some $t$, which implies $\vdash_p t\typ A$ for all $p$. 

Similarly to \cite{coquand2010note} we can state a conservativity result for this extension. Let $\vdash g\typ \tnat\rightarrow \tbool$. 
 We say that $g$ is compatible with a condition $p$ if $g$ is such that $\vdash g\,\xoverline{n} = b \typ \tbool$ whenever $(n,b) \in p$ and $\vdash g\,\xoverline{n} = \zer\typ \tbool$ otherwise. 
 We write $n_g$ for the smallest natural number such that $g\, \xoverline{n}_g = \zer$. Let $v_g \typ \neg \neg (\Sigma(x\typ \tnat)\,\iszero\,(g\,x))$ be the term given by $v_g\coloneqq \lambda x. x\,(\xoverline{n}_g,\zer)\typ \neg \neg (\Sigma(y\typ \tnat)\,\iszero\,(g\,y))$. 
 To see that $v_g$ is well typed, note that by design $\Gamma\vdash g\,\xoverline{n}_g = \zer \typ \tbool$ thus $\Gamma\vdash \iszero\,(g\,\xoverline{n}_g) = \tunit$ and $\Gamma \vdash (\xoverline{n}_g,\zer) \typ \Sigma(x\typ \tnat) \iszero\,(g\,x)$. 
 We have then $\Gamma, x\typ \neg (\Sigma(y\typ \tnat)\,\iszero\,(g\,y))\vdash x\,(\xoverline{n}_g,\zer)\typ \tempt$, thus $\Gamma\vdash \lambda x. x\, (\xoverline{n}_g,\zer) \typ \neg \neg (\Sigma(y\typ \tnat)\,\iszero\,(g\,y))$.
 

\begin{lem}[Conservativity]
Let $\vdash g \typ \tnat\rightarrow \tbool$ be compatible with some condition $p$. If $\Gamma \vdash_p J$ then $\Gamma[g/\f,v_g/\w] \vdash J[g/\f,v_g/\w]$, i.e.\ by replacing $\f$ with $g$ then $\w$ with $v_g$  we obtain a valid judgment in standard type theory. In particular, if we have $\Gamma\vdash_{\emptyset} J$ where neither $\f$ nor $\w$ occur in $\Gamma$ or $J$ then $\Gamma\vdash J$ is a valid judgment in standard type theory.
\end{lem}
\begin{proof}
We show that whenever the statement holds for the premise of a typing rule it holds for the conclusion. The statement will then follow by induction on the derivation tree of $\Gamma \vdash_p J$. 
 
 For the standard rules the proof is straightforward. For (\textsc{$\f$-eval}) we have $(\f\,\xoverline{n})[g/\f, v_g/\w] \coloneqq g\,\xoverline{n}$ and since $g$ is compatible with $p$ we have $\Gamma[g/\f, v_g/\w]\vdash g\,\xoverline{n}= p(n)\typ \tbool$ whenever $n\in \dom{p}$. For (\textsc{$\w$-term}) we have 
\begin{align*}
(\w\typ\neg\neg (\Sigma(x\typ \tnat)\,\iszero\,(\f\,x)))[g/\f, v_g/\w] &\coloneqq (\w\typ\neg\neg (\Sigma(x\typ \tnat)\, \iszero\,(g\,x)))[v_g/\w]\\
& \coloneqq v_g\typ\neg\neg (\Sigma(x\typ \tnat)\,\iszero \,(g\,x)). 
\end{align*}

For (\textsc{loc}) the statement follows from the observation that when $g$ is compatible with $p$ and $p\covt S$ then $g$ is compatible with exactly one $q \in S$. From $\Gamma \vdash_q J$ by IH we get $\Gamma[g/\f,v_g/\w] \vdash J[g/\f,v_g/\w]$.
\end{proof}

\section{A Semantics of the forcing extension}
\label{SEC:semantics}

In this section we outline a semantics for the forcing extension given in the previous section. We will interpret the judgments of type theory by computability predicates and relations defined by reducibility to computable weak head normal forms.

\subsection{Reduction rules}
We extend the $\beta, \iota$ conversion with $\f\,\xoverline{n} \red_p b$ whenever $(n,b) \in p$. To ease the presentation of the proofs and definitions we introduce {\em evaluation contexts} following \cite{WRIGHT199438}.
\begin{align*}
\mathbb{E}\Coloneqq & [\;] \mid \mathbb{E}\, u \mid \mathbb{E}\fpr \mid \mathbb{E}\spr \mid \Suc\,\mathbb{E} \mid \f\, \mathbb{E}\\
& \botrec (\lambda x. C)\,\mathbb{E} \mid \unitrec (\lambda x. C)\,a\,\mathbb{E}
\mid \boolrec (\lambda x. C)\, a_0\,a_1\,\mathbb{E}\mid \natrec (\lambda x. C)\, c_z\,g\,\mathbb{E}
\end{align*}

An expression $\mathbb{E}[e]$ is then the expression resulting from replacing the hole $[\;]$ by $e$. 

We have the following reduction rules:

\AxiomC{\vphantom{$\vdash$}}
\UnaryInfC{$\unitrec (\lambda x. C)\;c\;\zer \red c $}
\DisplayProof
\,
\AxiomC{$\vphantom{\vdash}$}
\UnaryInfC{$\boolrec (\lambda x. C)\,c_0\,c_1\,\zer \red c_0$}
\DisplayProof
\,
\AxiomC{$\vphantom{\vdash}$}
\UnaryInfC{$\boolrec (\lambda x. C)\,c_0\,c_1\,\one \red c_1$}
\DisplayProof\smallskip

\AxiomC{$\vphantom{\vdash}$}
\UnaryInfC{$\natrec (\lambda x. C)\,c_z\,g\,\zer \red c_z$}
\DisplayProof
\,
\AxiomC{$\vphantom{\vdash}$}
\UnaryInfC{$\natrec (\lambda x. C)\,c_z\,g\,(\Suc\,\xoverline{k}) \red g\,\xoverline{k}\, (\natrec (\lambda x. C)\,c_z\,g\,\xoverline{k})$}
\DisplayProof\medskip

\AxiomC{$\vphantom{\vdash}$}
\UnaryInfC{$(\lambda x. t)\, a \red t[a/x]$}
\DisplayProof
\,
\AxiomC{\vphantom{$\vdash$}}
\UnaryInfC{$(u,v)\fpr \red u$}
\DisplayProof
\,
\AxiomC{\vphantom{$\vdash$}}
\UnaryInfC{$(u,v)\spr \red v$}
\DisplayProof\medskip

\AxiomC{$e\red e'$}
\UnaryInfC{$e\red_p e'$}
\DisplayProof
\,
\AxiomC{$k \in \dom{p}$}
\LeftLabel{\scriptsize{\textsc{$\f$-red}}}
\UnaryInfC{$\f\,\xoverline{k} \red_p p(k)$}
\DisplayProof
\,
\AxiomC{$e\red_p e'$}
\UnaryInfC{$\mathbb{E}[e]\red_p \mathbb{E}[e']$}
\DisplayProof\medskip

Note that we reduce under $\Suc$. Also note that the relation $\red$ is monotone, that is if $q \leqslant p$ and $t \red_p u$ then $t \red_q u$. In the following we will show that $\red$ is also local, i.e.\ if $p\covt S$ and $t\red_q u$ for all $q\in S$ then $t\red_p u$. 

\begin{lem}
\label{untypedredislocal}
If $m\notin \dom{p}$ and $t\red_{p(m\mapsto \zer)} u$ and $t\red_{p(m\mapsto \one)} u$ then $t\red_p u$.
\end{lem}
\begin{proof}
By induction on the derivation of $t\red_{p(m\mapsto \zer)} u$. If  $t\red_{p(m\mapsto \zer)} u$ is derived by (\textsc{$\f$-red}) then $t\coloneqq \f\,\xoverline{k}$ and $u\coloneqq p(m\mapsto \zer)(k)$ for some $k\in \dom{p(m\mapsto \zer)}$. But since we also have a reduction $\f\,\xoverline{k}\red_{p(m\mapsto \one)} u$,  we have $p(m\mapsto \one)(k) \coloneqq u \coloneqq p(m\mapsto \zer)(k)$ which could only be the case if $k\in \dom{p}$. Thus we have a reduction $\f\,\xoverline{k} \red_p u\coloneqq p(k)$. If on the other hand we have a derivation $t\red u$, then we have $t\red_p u$ directly. If the derivation $t\red_{p(m\mapsto \zer)} u$ has the form $\mathbb{E}[e]\red_{p(m\mapsto \zer)} \mathbb{E}[e']$ then we have also $\mathbb{E}[e]\red_{p(m\mapsto \one)} \mathbb{E}[e']$. Hence, $e\red_{p(m\mapsto \zer)} e'$ and $e\red_{p(m\mapsto \one)} e'$. By IH $e\red_p e'$, thus $\mathbb{E}[e]\red_p \mathbb{E}[e']$.
\end{proof}

\begin{lem}
\label{reductionislevelbehaved}
Let $q\leqslant p$. If $t\red_q u$ then $t\red_p u$ or $t$ has the form $\mathbb{E}[\f\,\xoverline{m}]$ for some $m\in \dom{q}\setminus \dom{p}$.
\end{lem}
\begin{proof}
By induction on the derivation of $t\red_q u$. If the reduction $t\red_q u$ has the form $\f\,\xoverline{k} \red_q q(k)$ then either $k\notin\dom{p}$ and the statement follows or $k\in \dom{p}$ and we have $t\red_p u$. If on the other hand we have $t\red u$ then $t\red_p u$ immediately.
If $t\red_q u$ has the form $\mathbb{E}[e] \red_q \mathbb{E}[e']$ then $e\red_q e'$ and the statement follows by induction.
\end{proof}

\begin{cor}
\label{eitherreduceorstuck}
Let $t\red_{p(m\mapsto \zer)} u$ and $t\red_{p(m\mapsto \one)} v$ for some $m\notin\dom{p}$. If $u\coloneqq v$ then $t\red_p u$; otherwise, $t$ has the form $\mathbb{E}[\f\,\xoverline{m}]$.
\end{cor}

Define $p\vdash t\red u \typ A$ to mean $t\red_p u$ and $\vdash_p t=u \typ A$ and write $p\vdash A\red B$ for $p\vdash A\red B\typ \tuniv$. 

Note that it holds that if $p\vdash t\red u\typ \Pi(x\typ F) G$ and $\vdash a\typ F$ then $p\vdash t\,a\red u\,a \typ G[a]$ and if $p\vdash t\red u \typ \Sigma(x\typ F) G$ then $p\vdash t\fpr \red u\fpr \typ F$ and $p\vdash t\spr \red u\spr \typ G[t\fpr]$. 

We define a closure for this relation as follows:\medskip
\begin{center}
\AxiomC{$\vdash_p t\typ A$}
\UnaryInfC{$p\vdash t \rtred t \typ A$}
\DisplayProof
\,
\AxiomC{$p\vdash t \red u\typ A$}
\UnaryInfC{$p\vdash t \rtred u\typ A$}
\DisplayProof
\,
\AxiomC{$p\vdash t \red u \typ A$}
\AxiomC{$p\vdash u \rtred v \typ A$}
\BinaryInfC{$p\vdash t \rtred v\typ A$}
\DisplayProof\smallskip

\AxiomC{$\vdash_p A$}
\UnaryInfC{$p\vdash A\rtred A$}
\DisplayProof
\;\;\;
\AxiomC{$p\vdash A \red B$}
\UnaryInfC{$p\vdash A \rtred B$}
\DisplayProof
\;\;\;
\AxiomC{$p\vdash A \red B $}
\AxiomC{$p\vdash B \rtred C $}
\BinaryInfC{$p\vdash A \rtred C$}
\DisplayProof
\end{center}\smallskip


A term $t$ is {\em in} $p$-whnf if whenever $p\vdash t \rtred u\typ A$ then $t\coloneqq u$. 

A whnf is \emph{canonical} if it has one of the forms:
\begin{align*}
& \zer,\one,\xoverline{n}, \lambda x. t, (a,b), \f, \w, \tempt,\tunit,\tbool,\tnat,\tuniv,\Pi(x\typ F) G, \Sigma(x\typ F) G,\\
& \botrec (\lambda x. C), \unitrec (\lambda x. C)\,a, \boolrec (\lambda x. C)\, a_0\,a_1, \natrec (\lambda x. C)\, c_z\,g
\end{align*}

A $p$-whnf is \emph{proper} if it is canonical or it is of the form $\mathbb{E}[\f\,\xoverline{k}]$ for $k\notin\dom {p}$. 

A canonical $p$-whnf has no further reduction at any $q\leqslant p$. A non-canonical proper $p$-whnf, i.e.\ of the form $\mathbb{E}[\f\,\xoverline{k}]$ for $k\notin\dom{p}$, have further reduction at some $q\leqslant p$, namely when $k\in \dom{q}$.

We have the following corollaries to Lemma~\ref{untypedredislocal} and Corollary~\ref{eitherreduceorstuck}.

\begin{cor}
\label{typedreduceorstuck}
Let $m\notin \dom{p}$. Let $p(m\mapsto \zer) \vdash t\red u\typ A$ and $p(m\mapsto \one)\vdash t\red v\typ A$. If $u\coloneqq v$ then $p\vdash t\red u \typ A$; otherwise $t$ has the form  $\mathbb{E}[\f\,\xoverline{m}]$.
\end{cor}

\begin{cor}
\label{onestepredmonolocal}
Let $p\covt S$ and $q \vdash t\red u \typ A$ for all $q\in S$. We have $p\vdash t\red u \typ A$.
\end{cor}
\begin{proof}
By induction on $S$. If $S\coloneqq \{p\}$ the the statement follows. Assume the statement holds for $p(m\mapsto \zer)\covt S_0$ and $p(m\mapsto \one)\covt S_1$ and let $S\coloneqq S_0 \cup S_1$. By IH, $p(m\mapsto \zer)\vdash t\red u\typ A$ and $p(m\mapsto \one)\vdash t\red u\typ A$. From Lemma~\ref{untypedredislocal}, $t\red_p u$. Since $\vdash_{p(m\mapsto \zer)} t=u\typ A$ and $\vdash_{p(m\mapsto \one)} t=u\typ A$, then $\vdash_p t=u\typ A$. Thus $p\vdash t\red u\typ A$.
\end{proof}

Note that if $q\leqslant p$ and $p\vdash t\rtred u \typ A$ then $q\vdash t\rtred u\typ A$. However if $p(m\mapsto \zer)\vdash t\rtred u\typ A$ and $p(m\mapsto \one)\vdash t\rtred u\typ A$ it is not necessarily the case that $p\vdash t\rtred u\typ A$. E.g.\ we have that $\{(m,\zer)\}\vdash \boolrec (\lambda x. \tnat)\,\xoverline{n}\,\xoverline{n}\,(\f\,\xoverline{m}) \rtred \xoverline{n}\typ \tnat$ and $\{(m,\one)\}\vdash \boolrec (\lambda x. \tnat)\,\xoverline{n}\,\xoverline{n}\,(\f\,\xoverline{m}) \rtred \xoverline{n}\typ \tnat$ but it is \emph{not} true that $\emptyset\vdash \boolrec (\lambda x. \tnat) \,\xoverline{n}\,\xoverline{n}\,(\f\,\xoverline{m}) \rtred \xoverline{n}\typ \tnat$.

For a closed term $\vdash_p t\typ A$, we say that $t$ {\em has} a $p$-whnf if $p\vdash t\rtred  u \typ A$ and $u$ is in $p$-whnf. If $u$ is canonical, respectively proper, we say that $t$ has a canonical, respectively proper, $p$-whnf.  

Since the reduction relation is deterministic we have:

\begin{lem}
\label{LEM:uniquenessofwhnf}
A term $\vdash_p t\typ A$ has at most one $p$-whnf.
\end{lem}

\begin{cor}
\label{properwhnflocal}
Let $\vdash_p t\typ A$ and $m\notin \dom{p}$. If $t$ has proper $p(m\mapsto \zer)$-whnf and a proper $p(m\mapsto \one)$-whnf then $t$ has a proper $p$-whnf. 
\end{cor}
\begin{proof}
Let $p(m\mapsto \zer)\vdash t\rtred u\typ A$ and $p(m\mapsto \one)\vdash t\rtred v\typ A$ with $u$ in proper $p(m\mapsto \zer)$-whnf and $v$ in proper $p(m\mapsto \one)$-whnf. If $t\coloneqq u$ or $t\coloneqq v$ then $t$ is already in proper $p$-whnf. Alternatively we have reductions $p(m\mapsto \zer)\vdash t\red u_1\typ A$ and $p(m\mapsto \one)\vdash t\red v_1 \typ A$. By Corollary~\ref{typedreduceorstuck} either $t$ is in proper $p$-whnf or $u_1\coloneqq v_1$ and $p\vdash t\red u_1 \typ A$. It then follows by induction that $u_1$, and thus $t$, has a proper $p$-whnf.
\end{proof}

\subsection{Computability predicate and relation}
We define inductively a forcing relation $p\forces A$ to express that a type $A$ is computable at $p$. Mutually by recursion we define relations $p\forces a\typ A$ ($a$ computable of type $A$ at $p$), $p\forces A=B$ ($A$ and $B$ are computably equal at $p$), and $p\forces a=b \typ A$ ($a$ is computably equal to $b$ of type $A$ at $p$). 
The definition fits the generalized mutual induction-recursion schema \cite{dybjer_inductive-recursive}\footnote{However, for the canonical proof below we actually need something weaker than an
inductive-recursive definition (arbitrary fixed-point instead of {\em least} fixed-point), reflecting
the fact that the universe is defined in an open way \cite{martinloftypes}.}.

The following rules have an implicit (hidden) premise $\vdash_p A$\medskip
\begin{center}
\AxiomC{$p\vdash A \rtred \tempt$}
\LeftLabel{\scriptsize{$\lfor_\tempt$}}
\UnaryInfC{$p\forces A$}
\DisplayProof
\,
\AxiomC{$p\vdash A \rtred \tunit$}
\LeftLabel{\scriptsize{$\lfor_\tunit$}}
\UnaryInfC{$p\forces A$}
\DisplayProof
\,
\AxiomC{$p\vdash A \rtred \tbool$}
\LeftLabel{\scriptsize{$\lfor_\tbool$}}
\UnaryInfC{$p\forces A$}
\DisplayProof
\,
\AxiomC{$p\vdash A \rtred \tnat$}
\LeftLabel{\scriptsize{$\lfor_\tnat$}}
\UnaryInfC{$p\forces A$}
\DisplayProof\medskip
\,
\AxiomC{\vphantom{$p\vdash$}}
\LeftLabel{\scriptsize{$\lfor_\tuniv$}}
\UnaryInfC{$p\forces \tuniv$}
\DisplayProof

\begin{prooftree}
\alwaysNoLine
\AxiomC{$p \forces F$}
\AxiomC{$p\vdash A \rtred \Pi (x\typ F) G$}
\UnaryInfC{$\forall q \leqslant p (q\forces a\typ F\Rightarrow q\forces G[a])$}
\AxiomC{$\forall q \leqslant p (q\forces a=b\typ F\Rightarrow q\forces G[a]=G[b])$}
\alwaysSingleLine
\LeftLabel{\scriptsize{$\lfor_\Pi$}}
\TrinaryInfC{$p\forces A$}
\end{prooftree}

\begin{prooftree}
\alwaysNoLine
\AxiomC{$p \forces F$}
\AxiomC{$p\vdash A \rtred \Sigma (x\typ F) G$}
\UnaryInfC{$\forall q \leqslant p (q\forces a\typ F\Rightarrow q\forces G[a])$}
\AxiomC{$\forall q \leqslant p (q\forces a=b\typ F\Rightarrow q\forces G[a]=G[b])$}
\alwaysSingleLine
\LeftLabel{\scriptsize{$\lfor_\Sigma$}}
\TrinaryInfC{$p\forces A$}
\end{prooftree}

\AxiomC{$p\vdash A\rtred \mathbb{E}[\f\, \xoverline{k}]$}
\AxiomC{$k\notin \dom{p}$}
\AxiomC{$p(k\mapsto \zer)\forces A$} 
\AxiomC{$p(k\mapsto \one)\forces A$}
\LeftLabel{\scriptsize{$\lfor_\mathrm{Loc}$}}
\QuaternaryInfC{$p\forces A$}
\DisplayProof
\end{center}\bigskip

\begin{enumerate}
	\item Assuming $p\forces A$ by $\lfor_\tempt$
	\begin{enumerate}
		\item Assuming $p\forces B$ then $p\forces A=B$ if
		\begin{enumerate}
			\item $p\vdash B\rtred \tempt$.
			\item $p\vdash B\rtred \mathbb{E}[\f\,\xoverline{k}]$, $k\notin \dom{p}$ and $p(k\mapsto \zer)\forces A=B$ and $p(k\mapsto \one) \forces A=B$.
		\end{enumerate}
		\item $p \nforces t\typ A$ for all $t$.
		\item $p\nforces t=u\typ A$ for all $t$ and $u$.
	\end{enumerate}

	\item Assuming $p\forces A$ by $\lfor_\tunit$ 
	\begin{enumerate}
		\item Assuming $p\forces B$ then $p\forces A=B$ if
		\begin{enumerate}
			\item $p\vdash B\rtred \tunit$.
			\item $p\vdash B\rtred \mathbb{E}[\f\,\xoverline{k}]$, $k\notin \dom{p}$ and $p(k\mapsto \zer)\forces A=B$ and $p(k\mapsto \one) \forces A=B$.
		\end{enumerate}
		\item $p \forces t\typ A$ if
		\begin{enumerate}
			\item $p\vdash t\rtred \zer \typ A$
			\item $p\vdash t\rtred \mathbb{E}[\f\,\xoverline{k}]\typ A$, $k\notin\dom{p}$ and $p(k\mapsto \zer)\forces t\typ A$ and $p(k\mapsto \one) \forces t\typ A$.
		\end{enumerate}
		\item Assuming $p\forces t\typ A$ and $p\forces u\typ A$ then $p\forces t=u \typ A$ if
		\begin{enumerate}
			\item $p\vdash t\rtred \zer \typ A$ and $p\vdash u\rtred \zer \typ A$.
			\item $p\vdash t\rtred \zer \typ A$ and $p\vdash u\rtred \mathbb{E}[\f\,\xoverline{k}] \typ A, k\notin\dom{p}$ and $p(k\mapsto \zer)\forces t=u\typ A$ and $p(k\mapsto \one)\forces t= u\typ A$.
			\item $p\vdash t\rtred  \mathbb{E}[\f\,\xoverline{k}]\typ A, k\notin\dom{p}$ and $p(k\mapsto \zer)\forces t=u\typ A$ and $p(k\mapsto \one)\forces t= u\typ A$.
		\end{enumerate}
	\end{enumerate}

	\item Assuming $p\forces A$ by $\lfor_\tbool$ 
	\begin{enumerate}
		\item Assuming $p\forces B$ then $p\forces A=B$ if
		\begin{enumerate}
			\item $p\vdash B\rtred \tbool$.
			\item $p\vdash B\rtred \mathbb{E}[\f\,\xoverline{k}]$, $k\notin \dom{p}$ and $p(k\mapsto \zer)\forces A=B$ and $p(k\mapsto \one) \forces A=B$.
		\end{enumerate}
		\item $p \forces t\typ A$ if
		\begin{enumerate}
			\item $p\vdash t\rtred b \typ A$ for some $b\in \{\zer,\one\}$.
			\item $p\vdash t\rtred \mathbb{E}[\f\,\xoverline{k}]$, $k\notin\dom{p}$ and $p(k\mapsto \zer)\forces t\typ A$ and $p(k\mapsto \one) \forces t\typ A$.
		\end{enumerate}
		\item Assuming $p\forces t\typ A$ and $p\forces u\typ A$ then $p\forces t=u \typ A$ if
		\begin{enumerate}
			\item $p\vdash t\rtred b \typ A$ and $p\vdash u\rtred b \typ A$ for some $b\in \{\zer,\one\}$.
			\item $p\vdash t\rtred b \typ A$, $b\in \{\zer,\one\}$ and $p\vdash u\rtred \mathbb{E}[\f\,\xoverline{k}]\typ A, k\notin\dom{p}$ and $p(k\mapsto \zer)\forces t=u\typ A$ and $p(k\mapsto \one)\forces t= u\typ A$.
			\item $p\vdash t\rtred  \mathbb{E}[\f\,\xoverline{k}]\typ A, k\notin\dom{p}$ and $p(k\mapsto \zer)\forces t=u\typ A$ and $p(k\mapsto \one)\forces t= u\typ A$.
		\end{enumerate}
	\end{enumerate}
	
	\item Assuming $p\forces A$ by $\lfor_\tnat$ 
	\begin{enumerate}
		\item Assuming $p\forces B$ then $p\forces A=B$ if
		\begin{enumerate}
			\item $p\vdash B\rtred \tnat$.
			\item $p\vdash B\rtred \mathbb{E}[\f\,\xoverline{k}]$, $k\notin \dom{p}$ and $p(k\mapsto \zer)\forces A=B$ and $p(k\mapsto \one) \forces A=B$.
		\end{enumerate}
		\item $p \forces t\typ A$ if
		\begin{enumerate}
			\item $p\vdash t\rtred \xoverline{n} \typ A$.
			\item $p\vdash t\rtred \mathbb{E}[\f\,\xoverline{k}]$, $k\notin\dom{p}$ and $p(k\mapsto \zer)\forces t\typ A$ and $p(k\mapsto \one) \forces t\typ A$.
		\end{enumerate}
		\item Assuming $p\forces t\typ A$ and $p\forces u\typ A$ then $p\forces t=u \typ A$ if
		\begin{enumerate}
			\item $p\vdash t\rtred \xoverline{n} \typ A$ and $p\vdash u\rtred \xoverline{n} \typ A$.
			\item $p\vdash t\rtred \xoverline{n} \typ A$ and $p\vdash u\rtred \mathbb{E}[\f\,\xoverline{k}]\typ A, k\notin\dom{p}$ and $p(k\mapsto \zer)\forces t=u\typ A$ and $p(k\mapsto \one)\forces t= u\typ A$.
			\item $p\vdash t\rtred  \mathbb{E}[\f\,\xoverline{k}]\typ A, k\notin\dom{p}$ and $p(k\mapsto \zer)\forces t=u\typ A$ and $p(k\mapsto \one)\forces t= u\typ A$.
		\end{enumerate}
	\end{enumerate}

	\item Assuming $p\forces A$ by $\lfor_\Pi$ (Let $p\vdash A\rtred \Pi(x\typ F) G$). 
	\begin{enumerate}
		\item Assuming $p\forces B$ and $\vdash_p A=B$ then $p\forces A=B$ if
		\begin{enumerate}
			\item $p\vdash B\rtred \Pi(x\typ H) E$ and $p\forces F=H$ and $\forall q\leqslant p (q\forces a\typ F\Rightarrow q\forces G[a]=E[a])$.
			\item $p\vdash B\rtred \mathbb{E}[\f\,\xoverline{k}]$, $k\notin \dom{p}$ and $p(k\mapsto \zer)\forces A=B$ and $p(k\mapsto \one) \forces A=B$.
		\end{enumerate}
		\item Assuming $\vdash_p t\typ A$ then $p \forces t\typ A$ if
		\begin{enumerate}
			\item[] $\forall q\leqslant p(q\forces a\typ F \Rightarrow q\forces t\,a \typ G[a])$ and $\forall q\leqslant p (q\forces a=b\typ F\Rightarrow q\forces t\,a = t\,b \typ G[a])$. 
		\end{enumerate}
		\item Assuming $p\forces t\typ A$ and $p\forces u\typ A$ and $\vdash_p t=u\typ A$ then $p\forces t=u \typ A$ if
		\begin{enumerate}
			\item[] $\forall q \leqslant p (q\forces a\typ F \Rightarrow q\forces t\, a = u\, a \typ G[a])$
		\end{enumerate}
	\end{enumerate}

	\item Assuming $p\forces A$ by $\lfor_\Sigma$ (Let $p\vdash A\rtred \Sigma(x\typ F) G$). 
	\begin{enumerate}
		\item Assuming $p\forces B$ and $\vdash_p A=B$ then $p\forces A=B$ if
		\begin{enumerate}
			\item $p\vdash B\rtred \Sigma(x\typ H) E$ and $p\forces F=H$ and $\forall q\leqslant p (q\forces a\typ F\Rightarrow q\forces G[a]=E[a])$.
			\item $p\vdash B\rtred \mathbb{E}[\f\,\xoverline{k}]$, $k\notin \dom{p}$ and $p(k\mapsto \zer)\forces A=B$ and $p(k\mapsto \one) \forces A=B$.
		\end{enumerate}
		\item Assuming $\vdash_p t\typ A$ then $p \forces t\typ A$ if
		\begin{enumerate}
			\item[] $p\forces t\fpr \typ F$ and $p\forces t\spr \typ G[t\fpr]$
		\end{enumerate}
		\item Assuming $p\forces t\typ A$ and $p\forces u\typ A$ and $\vdash_p t=u\typ A$ then $p\forces t=u \typ A$ if
		\begin{enumerate}
			\item[] $p\forces t\fpr = u\fpr \typ F$ and $p\forces t\spr = u\spr \typ G[t\fpr]$
		\end{enumerate}
	\end{enumerate}

	\item Assuming $p\forces A$ by $\lfor_\tuniv$ (i.e.\ $A \coloneqq \tuniv$).
	\begin{enumerate}
		\item Assuming $p\forces B$ then $p\forces A=B$ if $B\coloneqq \tuniv$
		\item Assuming $\vdash_p L\typ A$ then $p \forces L\typ A$ if
		\begin{enumerate}
			\item $p\vdash L\rtred M$ with $M\in \{\tempt, \tunit, \tbool, \tnat\}$
			\item $p\vdash L\rtred \Pi(x\typ F) G$ and $p\forces F\typ A$ and\\ $\forall q\leqslant p (q\forces a\typ F \Rightarrow q\forces G[a] \typ A)$ and $\forall q\leqslant p (q\forces a=b\typ F \Rightarrow q\forces G[a]=G[b] \typ A)$.
			\item $p\vdash L\rtred \Sigma(x\typ F) G$ and $p\forces F\typ A$ and\\ $\forall q\leqslant p (q\forces a\typ F \Rightarrow q\forces G[a] \typ A)$ and $\forall q\leqslant p (q\forces a=b\typ F \Rightarrow q\forces G[a]=G[b] \typ A)$
			\item $p\vdash L \rtred \mathbb{E}[\f\,\xoverline{k}]$, $k\notin\dom{p}$ and $p(k\mapsto \zer)\forces L\typ A$ and $p(k\mapsto \one) \forces L\typ A$.
		\end{enumerate}
		\item Assuming $p\forces L\typ A$ and $p\forces P\typ A$ and $\vdash_p L=P\typ A$ then $p\forces L=P \typ A$ if
		\begin{enumerate}
			\item $p\vdash L\rtred M$ and $p\vdash P\rtred M$ for $M\in \{\tempt, \tunit, \tbool,\tnat\}$.
			\item $p\vdash L\rtred \Pi(x\typ F) G$ and $p\vdash P\rtred \Pi(x\typ H) E$ and\\ $p\forces F=H\typ A$ and $\forall q\leqslant p (q\forces a\typ F \Rightarrow q\forces G[a] = E[a] \typ A)$
			\item $p\vdash L\rtred \Sigma(x\typ F) G$ and $p\vdash P\rtred \Sigma(x\typ H) E$ and\\ $p\forces F=H\typ A$ and $\forall q\leqslant p (q\forces a\typ F \Rightarrow q\forces G[a] = E[a] \typ A)$
			\item $p\vdash L\rtred M$ with $M$ in canonical $p$-whnf 
			and $p\vdash P\rtred \mathbb{E}[\f\,\xoverline{k}], k\notin\dom{p}$ and $p(k\mapsto \zer)\forces L=P\typ A$ and $(k\mapsto \one) \forces L=P\typ A$.
			\item $p\vdash L\rtred \mathbb{E}[\f\,\xoverline{k}], k\notin\dom{p}$ and $p(k\mapsto \zer)\forces L=P\typ A$ and $(k\mapsto \one) \forces L=P\typ A$.
		\end{enumerate}
	\end{enumerate}

	\item Assuming $p\forces A$ by $\lfor_\mathrm{Loc}$ (i.e.\ $p\vdash A\rtred \mathbb{E}[\f\,\xoverline{k}], k\notin\dom{p}$ and $p(k\mapsto \zer)\forces A$ and $p(k\mapsto \one)\forces A$).
	\begin{enumerate}
		\item Assuming $p\forces B$ and $\vdash_p A= B$ then $p\forces A=B$ if $p(k\mapsto \zer)\forces A=B$ and $p(k\mapsto \one) \forces A=B$.
		\item Assuming $\vdash_p t\typ A$ then $p \forces t\typ A$ if $p(k\mapsto \zer)\forces t\typ A$ and $p(k\mapsto \one) \forces t\typ A$.
		\item Assuming $p\forces t\typ A$ and $p\forces u\typ A$ and $\vdash_p t=u\typ A$ then $p\forces t=u \typ A$ if $p(k\mapsto \zer)\forces t=u\typ A$ and $p(k\mapsto \one) \forces t=u\typ A$.
	\end{enumerate}

\end{enumerate}

Note from the definition that when $p\forces A=B$ then $p\forces A$ and $p\forces B$, when $p\forces a\typ A$ then $p\forces A$ and when $p\forces a=b \typ A$ then $p\forces a\typ A$ and $p\forces b\typ A$. It follows also from the definition that $\vdash_p J$ whenever $p\forces J$. 

The clause ($\mathrm{F_{Loc}}$) gives semantics to \emph{variable types}. For example,  if $p\coloneqq\{(0,\zer)\}$ and $q\coloneqq\{(0,\one)\}$ the type $R\coloneqq \boolrec (\lambda x. \tuniv)\,\tunit\,\tnat\,(\f\,\zer)$ has reductions $p\vdash R \rtred \tunit$ and $q\vdash R\rtred \tnat$. Thus $p\forces R$ and $q\forces R$ and since $\emptyset \covt \{p,q\}$ we have $\emptyset\forces R$. 

\begin{lem}
\label{forcedtypeslocallyreducetocanonical}
If $p\forces A$ then there is a partition $p\covt S$ where $A$ has a canonical $q$-whnf for all $q\in S$. If $p\forces A=B$ then there is a partition $p\covt S$ where $A$ and $B$ have similar canonical $q$-whnf for all $q\in S$, i.e.\ $q\vdash A\rtred A'$ and $q\vdash B\rtred B'$ where $(A',B')$ is of the form $(\tempt,\tempt)$, $(\tunit, \tunit)$, $(\tbool, \tbool)$, $(\tnat, \tnat)$, $(\tuniv, \tuniv)$, $(\Pi(x\typ F) G, \Pi(x\typ H) E)$, or $(\Sigma(x\typ F) G, \Sigma(x\typ H) E)$.
\end{lem}
\begin{proof}
The statement follows from the definition by induction on the derivation of $p\forces A$
\end{proof}

\begin{cor}
\label{existenceofpwhnfislocal}
Let $p\covt S$. If $q\forces A$ for all $q\in S$ then $A$ has a proper $p$-whnf. 
\end{cor}
\begin{proof}
Follows from Lemma~\ref{forcedtypeslocallyreducetocanonical} and Corollary~\ref{properwhnflocal} by induction.
\end{proof}

\begin{lem}
	\label{LEM:BaseTypeCanonicity}
	Let $p\vdash A\rtred M$ with $M \in \{\tunit,\tbool,\tnat\}$. If $p\forces t\typ A$ then there is a partition $p\covt S$ where $t$ has a canonical $q$-whnf for all $q\in S$. If $p\forces t=u\typ A$ then there is a partition $p\covt S$ where $t$ and $u$ have the same canonical $q$-whnf for each $q\in S$.
\end{lem}
\begin{proof}
Follows from the definition.
\end{proof}

The rest of this section is dedicated to proving the following theorem:
\begin{thm}
\label{properties}
The following hold for the forcing relation
\begin{enumerate}
\item Monotonicity: If $q \leqslant p$ and $p\forces J$ then $q\forces J$.
\item Locality: If $p\covt S$ and $q\forces J$ for all $q\in S$ then $p\forces J$.
\item Reflexivity: If $p\forces A$ then $p\forces A=A$ and if $p\forces a\typ A$ then $p\forces a=a\typ A$.
\item Symmetry: If $p\forces A=B$ then $p\forces B=A$ and if $p\forces a=b\typ A$ then $p\forces b=a\typ A$.
\item Transitivity: If $p\forces A=B$ and $p\forces B=C$ then $p\forces A=C$ and if $p\forces a=b \typ A$ and $p\forces b=c \typ A$ then $p\forces a=c\typ A$.
\item Extensionality: If $p\forces A=B$ then if $p\forces a\typ A$ then $p\forces a\typ B$ and if $p\forces a=b\typ A$ then $p\forces a=b \typ B$.
\end{enumerate}
\end{thm}

In the premise of any forcing $p\forces J$ there are a number of typing judgments. Since the type system satisfy the properties listed in Theorem~\ref{properties} we will largely ignore these typing judgments in the proofs.

\begin{lem}\label{LEM:typeIntroMonotone}
If $p\forces A$ and $q\leqslant p$ then $q\forces A$.
\end{lem}
\begin{proof}
Let $p\forces A$ and $q\leqslant p$. By induction on the derivation of $p\forces A$ 
	\begin{enumerate}
		\item (Derivation by $\for_\tnat$)
		 Let $p\vdash A\rtred \tnat$. Since the reduction is monotone $q\vdash A\rtred \tnat$, thus $q\forces A$. \\
		 The statement follows similarly when $p\forces A$ holds by $\lfor_\tempt$, $\lfor_\tunit$, $\lfor_\tbool$, $\lfor_\tnat$ or $\lfor_\tuniv$.
		\item\label{typeIntroMonotonePi} (Derivation by $\lfor_\Pi$.)
		 Let $p\vdash A\rtred\Pi(x\typ F) G$. From the premise $p\forces F$, by IH, it follows that $q\forces F$. From $\forall r\leqslant p (r\forces a\typ F\Rightarrow r\forces G[a])$ and $\forall r\leqslant p (r\forces a=b\typ F\Rightarrow r\forces G[a]=G[b])$ it follows directly that $\forall s\leqslant q (s\forces a\typ F\Rightarrow s\forces G[a])$ and $\forall s\leqslant q (s\forces a=b\typ F\Rightarrow s\forces G[a]=G[b])$. Hence $q\forces A$.\\
		 The statement follows similarly when $p\forces A$ holds by $\lfor_\Sigma$
		\item (Derivation by $\lfor_\mathrm{Loc}$.) 
		 Let $p\vdash A\rtred \mathbb{E}[\f\,\xoverline{m}], m\notin \dom{p}$. If $m\in \dom{q}$ then $q\leqslant p(m\mapsto b)$ for some $b\in \{\zer,\one\}$. Since $p(m\mapsto b)\forces A$ with a derivation strictly smaller than the derivation of $p\forces A$ then by IH $q\forces A$. Alternatively, $q\vdash A\rtred \mathbb{E}[\f\,\xoverline{m}]$ but then $q(m\mapsto b)\leqslant p(m\mapsto b)$. By IH we have $q(m\mapsto \zer)\forces A$ and  $q(m\mapsto \one)\forces A$ and thus $q\forces A$.\qedhere
	\end{enumerate}
\end{proof}

\begin{lem}\label{LEM:typeEqMonotone}
If $p\forces A=B$ and $q\leqslant p$ then $q\forces A=B$.
\end{lem}
\begin{proof}
Let $p\forces A=B$ and $q\leqslant p$. We have then that $p\forces A$ and $p\forces B$. By Lemma~\ref{LEM:typeIntroMonotone} we have that $q\forces A$ and $q\forces B$. By induction on the derivation of $p\forces A$
\begin{enumerate}
	 \item\label{floctypeq} Let $p\forces A$ by $\lfor_\mathrm{Loc}$, i.e.\ $p\vdash A\rtred \mathbb{E}[\f\,\xoverline{k}], k\notin\dom{p}$ and $p(k\mapsto \zer)\forces A=B$ and $p(k\mapsto \one)\forces A=B$.  By induction on the derivation of $p\forces A=B$. If $k\in\dom{q}$ then $q\leqslant p(k\mapsto b)$ for some $b\in \{\zer,\one\}$. Since the derivation of $p(k\mapsto b)\forces A=B$ is strictly smaller than that of $p \forces A=B$,  by IH $q\forces A=B$. Otherwise, $k\notin \dom{q}$ and $q\vdash A\rtred \mathbb{E}[\f\,\xoverline{k}]$ and since $q(k\mapsto b)\leqslant p(k\mapsto b)$, by IH, $q(k\mapsto \zer)\forces A=B$ and $q(k\mapsto \one)\forces A=B$. By the definition $q\forces A=B$.  
	\item Let $p\forces A$ by $\lfor_\tnat$ (i.e.\ $p\vdash A\rtred \tnat$). By induction on the derivation of $p\forces A=B$ 
	\begin{enumerate}
		\item Let $p\vdash B\rtred \tnat$. We have directly that $q\forces A=B$ by monotonicity of the reduction.
		\item Let $p\vdash B\rtred \mathbb{E}[\f\,\xoverline{k}], k\notin\dom{p}$ and $p(k\mapsto \zer)\forces A=B$ and $p(k\mapsto \one)\forces A=B$. The statement then follows similarly to (\ref{floctypeq}). 
	\end{enumerate}
	The statement follows similarly when $p\forces A$ holds by $\lfor_\tempt, \lfor_\tunit, \lfor_\tbool$.
	\item Let $p\vdash A\rtred \Pi(x\typ F) G$. By induction on the derivation of $p\forces A=B$ 
	\begin{enumerate}
		\item Let $p\vdash B\rtred \Pi(x\typ H E)$ and $p\forces F=H$ and $\forall r\leqslant p(r\forces a\typ F\Rightarrow r\forces G[a]=E[a])$. By IH $q\forces F=H$. Directly we have $\forall s\leqslant q(s\forces a\typ F\Rightarrow s\forces G[a]=E[a])$. Thus $q\forces A=B$.
		\item Let $p\vdash B\rtred \mathbb{E}[\f\,\xoverline{k}], k\notin\dom{p}$ and $p(k\mapsto \zer)\forces A=B$ and $p(k\mapsto \one)\forces A=B$. The statement then follows similarly to (\ref{floctypeq}). 
	\end{enumerate}
	The statement follows similarly when $p\forces A$ holds by $\lfor_\Sigma$.
\item (Derivation by $\lfor_\tuniv$) We have then that $B\coloneqq \tuniv$ and thus $q\forces A=B$.\qedhere
\end{enumerate}
\end{proof}

\begin{lem}\label{LEM:termIntroMonotone}
If $p\forces t\typ A$ and $q\leqslant p$ then $q\forces t\typ A$.
\end{lem}
\begin{proof}
Let $p\forces t\typ A$ and $q\leqslant p$. From the definition we have that $p\forces A$. From Lemma \ref{LEM:typeIntroMonotone} we have that $q\forces A$. By induction on the derivation of $p\forces A$.
		\begin{enumerate}
			\item Let $p\forces A$ by $\lfor_\tnat$.  By induction on the derivation of $p\forces t\typ A$. 
			\begin{enumerate}
				\item Let $p\vdash t\rtred \xoverline{n}\typ A$. Then $q\vdash t\rtred \xoverline{n}\typ A$ and $q\forces t\typ A$.
				\item Let $p\vdash t\rtred \mathbb{E}[\f\,\xoverline{k}]\typ A, k\notin\dom{p}$ and $p(k\mapsto \zer)\forces t\typ A$ and $p(k\mapsto \one)\forces t\typ A$. If $k\in \dom{q}$ then $q\leqslant p(k\mapsto b)$ for some $b\in \{\zer,\one\}$, since the derivation of $p(k\mapsto b) \forces t\typ A$ is strictly smaller than the derivation $p\forces t\typ A$, by IH,  $q\forces t\typ A$. Otherwise, $k\notin\dom{q}$ and $q\vdash t\rtred \mathbb{E}[\f\,\xoverline{k}] \typ A$. But $q(k\mapsto b) \leqslant p(k\mapsto b)$ and by IH $q(k\mapsto \zer)\forces t\typ A$ and $q(k\mapsto \one)\forces t\typ A$. By the definition $q\forces t\typ A$. 
			\end{enumerate}
			The statement follows similarly when $p\forces A$ is derived by $\lfor_\tunit$ or $\lfor_\tbool$.
			
			\item Let $p\forces A$ by $\lfor_\tuniv$. The statement follows similarly to Lemma~\ref{LEM:typeIntroMonotone}.
			\item Let $p\forces A$ by $\lfor_\Pi$ and let $p\vdash A\rtred \Pi(x\typ F) G$. From $\forall r\leqslant p (r\forces a\typ F \Rightarrow r\forces t\,a \typ G[a])$ and $\forall r\leqslant p (r\forces a=b\typ F\Rightarrow r\forces t\,a=t\,b\typ G[a])$ we have directly that $\forall s\leqslant q (s\forces a\typ F \Rightarrow s\forces t\,a \typ G[a])$ and $\forall s\leqslant q (s\forces a=b\typ F\Rightarrow s\forces t\,a=t\,b\typ G[a])$. Thus $q\forces t\typ A$. 
			\item Let $p\forces A$ by $\lfor_\Sigma$ and let $p\vdash A\rtred \Sigma(x\typ F) G$. By induction on the derivation of $p\forces t\typ A$. From $p\forces t\fpr \typ F$ and $p\forces t\spr \typ G[t\fpr]$, by IH, $q\forces t\fpr \typ F$ and $q\forces t\spr \typ G[t\fpr]$, thus $q\forces t\typ A$.
			\item Let $p\forces A$ by $\lfor_\mathrm{Loc}$ and let $p\vdash A\rtred \mathbb{E}[\f\,\xoverline{k}], k\notin \dom {p}$ and $p(k\mapsto \zer)\forces A$ and $p(k\mapsto \one)\forces A$. By induction on the derivation of $p\forces t\typ A$. If $k\in\dom{q}$ then $q\leqslant p(k\mapsto b)$ for some $b\in \{\zer,\one\}$ from the premise $p(k\mapsto \zer)\forces t\typ A$ and $p(k\mapsto \one)\forces t\typ A$, by IH, $q\forces t\typ A$. Otherwise, $k\notin\dom{q}$ and $q\vdash A\rtred \mathbb{E}[\f\,\xoverline{k}]$ (i.e.\ $q\forces A$ by $\lfor_\mathrm{Loc}$). Since $q(k\mapsto b)\leqslant p(k\mapsto b)$, by the IH, $q(k\mapsto \zer)\forces t\typ A$ and $q(k\mapsto \one)\forces t\typ A$. By the definition $q\forces t\typ A$. \qedhere
		\end{enumerate}
\end{proof}

\begin{lem}\label{LEM:termEqMonotone}
 If $p\forces t=u\typ A$ and $q\leqslant p$ then $q\forces t=u\typ A$.
\end{lem}
\begin{proof}
Let $p\forces t=u \typ A$ and $q\leqslant p$. We have then that $p\forces A$, $p\forces t\typ A$, and $p\forces u\typ A$. By Lemma~\ref{LEM:typeIntroMonotone} $q\forces A$. By Lemma~\ref{LEM:termIntroMonotone} $q\forces t \typ A$ and $q\forces u \typ A$. By induction on the derivation $p\forces A$.

\begin{enumerate}
			\item Let $p\forces A$ by $\lfor_\tnat$.  By induction on the derivation of $p\forces t=u\typ A$.  
			\begin{enumerate}
				\item Let $p\vdash t\rtred \xoverline{n}\typ A$ and $p\vdash u\rtred \xoverline{n}\typ A$. By monotonicity of reduction $q\forces t=u\typ A$.
			\item \label{localtermeqforcingismonotone} Let $p\vdash t\rtred \xoverline{n}\typ A$ and $p\vdash t\rtred \mathbb{E}[\f\,\xoverline{k}]\typ A, k\notin\dom{p}$ and $p(k\mapsto \zer)\forces t=u\typ A$ and $p(k\mapsto \one)\forces t=u\typ A$. If $k\in \dom{q}$ then $q\leqslant p(k\mapsto b)$ for some $b\in \{\zer,\one\}$, by IH,  $q\forces t=u\typ A$. Otherwise,  $q\vdash t\rtred \mathbb{E}[\f\,\xoverline{k}]\typ A$. But $q(k\mapsto b) \leqslant p(k\mapsto b)$ and by IH $q(k\mapsto \zer)\forces t=u\typ A$ and $q(k\mapsto \one) \forces t=u\typ A$. By the definition $q\forces t=u\typ A$. 
			\item Let $p\vdash t\rtred \mathbb{E}[\f\,\xoverline{k}]\typ A, k\notin\dom{p}$. The statement follows similarly to (\ref{localtermeqforcingismonotone}).
			\end{enumerate}
			The statement follows similarly for when $p\forces A$ holds by $\lfor_\tunit$ or $\lfor_\tbool$.
			\item Let $p\forces A$ by $\lfor_\tuniv$. The statement follows by a proof similar to that of Lemma~\ref{LEM:typeEqMonotone}.
			\item Let $p\forces A$ by $\lfor_\Pi$ and let $p\vdash A\rtred \Pi(x\typ F) G$. From $\forall r\leqslant q(r\forces a\typ F\Rightarrow r\forces t\, a=u\, a\typ G[a])$ we have directly that $\forall s\leqslant q(s\forces a\typ F\Rightarrow s\forces t\, a=u\, a\typ G[a])$. Hence $q\forces t=u\typ A$.
			\item Let $p\forces A$ by $\lfor_\Sigma$ and let $p\vdash A\rtred \Sigma(x\typ F) G$. By induction on the derivation of $p\forces t=u\typ A$. From $p\forces t\fpr =u\fpr\typ F$ and $p\forces t\spr =u\spr\typ G[t\fpr]$, by IH we have $q\forces t\fpr =u\fpr\typ F$ and $q\forces t\spr =u\spr\typ G[t\fpr]$, thus $q\forces t=u\typ A$.
			\item Let $p\forces A$ by $\lfor_\mathrm{Loc}$ and let $p\vdash A\rtred \mathbb{E}[\f\,\xoverline{k}], k \notin\dom{p}$. By induction on the derivation of $p\forces t=u \typ A$. If $k\in \dom{q}$ then the statement follows by IH. If $k\notin\dom{q}$ then $q\vdash A\rtred \mathbb{E}[\f\,\xoverline{k}]$ (i.e.\ $q\forces A$ by $\lfor_\mathrm{Loc}$) and since $q(k\mapsto b)\leqslant p(k \mapsto b)$, by IH, $q(k\mapsto \zer)\forces t=u\typ A$ and $q(k\mapsto \one)\forces t=u\typ A$. Hence $q\forces t=u\typ A$. \qedhere
\end{enumerate}
\end{proof}

We collect the results of Lemmas~\ref{LEM:typeIntroMonotone}, \ref{LEM:termIntroMonotone}, \ref{LEM:termEqMonotone}, and \ref{LEM:typeEqMonotone} in the following corollary.
\begin{cor}[Monotonicity]
\label{COR:monotonicity}
If $p\forces J$ and $q\leqslant p$ then $q\forces J$.
\end{cor}

We write $\forces J$ when $\emptyset \forces J$. By monotonicity $\forces J$ iff $p\forces J$ for all $p$. 

\begin{lem}
\label{LEM:localitysublemma1}
Let $p\forces A$ and $p\forces B$. If $p(m\mapsto \zer)\forces A=B$ and $p(m\mapsto \one)\forces A=B$ then $p\forces A=B$.
\end{lem}
\begin{proof}
By induction on the derivation of $p\forces A$.
\begin{enumerate}
	\item Let $p\forces A$ by $\lfor_\tnat$. By induction on the derivation of $p\forces B$
	\begin{enumerate}
		\item If $p\forces B$ by $\lfor_\tnat$ then $p\forces A=B$ immediately.
		\item \label{typeqrightloc} If $p\forces B$ by $\lfor_\mathrm{Loc}$. 
		The statement follows similarly to~(\ref{loctypeqloc}) below. 
	\end{enumerate}
	The statement follows similarly when $p\forces A$ is derived by $\lfor_\tempt, \lfor_\tunit$ and $\lfor_\tbool$.
	\item Let $p\forces A$ by $\lfor_\Pi$ and let $p\vdash A\rtred \Pi(x\typ F) G$. By induction on the derivation of $p\forces B$
	\begin{enumerate}
		\item Let $p\forces B$ by $\lfor_\Pi$ and let $p\vdash B\rtred \Pi(x\typ H) E$. Since $p\forces A$ and $p\forces B$ we have $p\forces F$ and $p\forces H$. From the premise $p(m\mapsto \zer)\forces F=H$ and $p(m\mapsto \one)\forces F=H$ and by IH $p\forces F=H$.
		
		Let $q\leqslant p$ and $q\forces a\typ F$. If $q\leqslant p(m\mapsto b)$ for some $b\in \{\zer,\one\}$ then $q\forces G[a]=E[a]$. Otherwise, since $q(m\mapsto b)\leqslant p(m\mapsto b)$, by monotonicity $q(m\mapsto \zer) \forces G[a]=E[a]$ and $q(m\mapsto \one)\forces G[a]=E[a]$. From $p\forces A$ we have that $q\forces G[a]$ and from $p\forces B$ we have that $q\forces E[a]$. By IH $q\forces G[a] = E[a]$. We thus have $p\forces A=B$.
		\item Let $p\forces B$ by $\lfor_\mathrm{Loc}$. The statement then follows similarly to (\ref{loctypeqloc}) below.
	\end{enumerate}
	The statement follows similarly when $p\forces A$ is derived by $\lfor_\Sigma$.
	\item If $p\forces A$ by $\lfor_\tuniv$ then $A\coloneqq B\coloneqq \tuniv$ and $p\forces A=B$.
	\item \label{loctypeqloc} If $p\forces A$ by $\lfor_\mathrm{Loc}$. Let $p\vdash A\rtred \mathbb{E}[\f\,\xoverline{k}], k\notin\dom{p}$ and $p(k\mapsto \zer)\forces A$ and $p(k\mapsto \one)\forces A$. 
	\begin{enumerate}
		\item If $k=m$ then we have $p\forces A=B$ by the definition.
		\item If $k\neq m$. By monotonicity $p(k\mapsto b)\forces A$ and $p(k\mapsto b)\forces B$ and $p(k \mapsto b) (m\mapsto \zer) \forces A=B$ and $p(k\mapsto b) (m\mapsto \one) \forces A=B$ for all $b\in \{\zer,\one\}$. The derivation of $p(k\mapsto b)\forces A$ is strictly smaller than the derivation of $p\forces A$. By IH we have $p(k\mapsto \zer) \forces A=B$ and $p(k\mapsto \one)\forces A=B$. By the definition $p\forces A=B$.\qedhere
	\end{enumerate}
\end{enumerate}
\end{proof}

\begin{lem}\leavevmode
\label{LEM:typeIntroLocal}
If $p(m\mapsto \zer)\forces A$ and $p(m\mapsto \one)\forces A$ for some $m\notin \dom{p}$ then $p\forces A$.
\end{lem}
\begin{proof}
The proof is by by induction on the derivations $p(m\mapsto \zer)\forces A$. Note that from $p(m\mapsto \zer)\forces A$ and $p(m\mapsto \one)\forces A$ we have that $A$ has proper $p(m\mapsto \zer)$-whnf and $p(m\mapsto \one)$-whnf and by Corollary~\ref{existenceofpwhnfislocal} $A$ has a proper $p$-whnf.

\begin{enumerate}
	\item Let $p(m\mapsto \zer)\forces A$ by $\lfor_\tnat$
	\begin{enumerate}
			\item If $A$ has a canonical $p$-whnf then $p\vdash A\rtred \tnat$ and $p\forces A$.
			\item Otherwise, $p\vdash A\rtred \mathbb{E}[\xoverline{k}], k\notin\dom{p}$. Since $A$ has a canonical $p(m\mapsto \zer)$-whnf we have that $k=m$ and by the definition we have $p\forces A$ by $\lfor_\mathrm{Loc}$
	\end{enumerate}
	The statement follows similarly when $p(m\mapsto \zer) \forces A$ holds by $\lfor_\tempt$, $\lfor_\tunit$ or $\lfor_\tbool$.
	\item Let $p(m\mapsto \zer)\forces A$ by $\lfor_\Pi$
	\begin{enumerate}
		\item If $A$ has a canonical $p$-whnf then $p\vdash A\rtred \Pi(x\typ F) G$. From $p(m\mapsto 0)\forces A$ and $p(m\mapsto 1)\forces A$ we have $p(m\mapsto \zer)\forces F$ and $p(m\mapsto \one)\forces F$. By IH $p\forces F$. Let $q\leqslant p$
		\begin{enumerate}
			\item If $m\in \dom{q}$ then $q\leqslant p(m\mapsto b)$ for some $b\in \{\zer,\one\}$, then $(q\forces a\typ F \Rightarrow q\forces G[a])$ and $(q\forces a=b\typ F\Rightarrow q\forces G[a]=G[b])$.
			\item If $m\notin \dom{q}$ then $q\covt \{q(m\mapsto \zer), q(m\mapsto \one)\}$. Let $q\forces a\typ F$, by monotonicity $q(m\mapsto \zer) \forces a\typ F$ and $q(m\mapsto \one)\forces a\typ F$. Since $q(m\mapsto b)\leqslant p(m\mapsto b)$, by the definition $q(m\mapsto \zer) \forces G[a]$ and $q(m\mapsto \one)\forces G[a]$ and by IH $q\forces G[a]$. 
			
			Let $q\forces a=b\typ F$, by monotonicity $q(m\mapsto \zer)\forces a=b\typ F$ and $q(m\mapsto \one)\forces a=b\typ F$. But then $q(m\mapsto \zer)\forces G[a]$ and $q(m\mapsto \one)\forces G[a]$ and $q(m\mapsto \zer)\forces G[b]$ and $q(m\mapsto \one)\forces G[b]$ and $q(m\mapsto \zer)\forces G[a]=G[b]$ and $q(m\mapsto \one)\forces G[a]=G[b]$. By Lemma~\ref{LEM:localitysublemma1} we have $q\forces G[a]=G[b]$.
				\end{enumerate}
		\end{enumerate}
		The statement follows similarly when $p(m\mapsto \zer) \forces A$ holds by $\lfor_\Sigma$.
	\item Let $p(m\mapsto \zer)\forces A$ by $\lfor_\tuniv$ then $A\coloneqq \tuniv$ and $p\forces A$.
	\item Let $p(m\mapsto \zer)\forces A$ by $\lfor_\mathrm{Loc}$. Since $A$ doesn't have a canonical $p(m\mapsto \zer)$-whnf $A$ doesn't have a canonical $p$-whnf. Since $A$ has a proper $p$-whnf we have $p\vdash A\rtred \mathbb{E}[\f\,\xoverline{k}], k\notin\dom{p}$. 	
	\begin{enumerate}
	\item If $k=m$ then by the definition we have $p\forces A$ by $\lfor_\mathrm{Loc}$
	\item If $k\neq m$ then $p(m\mapsto b)\vdash A\rtred \mathbb{E}[\f\,\xoverline{k}]$. Hence $p(m\mapsto \one)\forces A$ by $\lfor_\mathrm{Loc}$. We have then $p(m\mapsto \zer)(k\mapsto \zer)\forces A$ and $p(m\mapsto \zer)(k\mapsto \one)\forces A$ and $p(m\mapsto \one)(k\mapsto \zer)\forces A$ and $p(m\mapsto \one)(k\mapsto \one)\forces A$. By IH $p(k\mapsto \zer) \forces A$ and $p(k\mapsto \one)\forces A$. By the definition $p\forces A$.\qedhere
	\end{enumerate}
	\end{enumerate}
\end{proof}

\begin{lem}
\label{LEM:typeEqLocal}
If $p(m\mapsto \zer)\forces A=B$ and $p(m\mapsto \one)\forces A=B$ for $m\notin\dom{p}$ then $p\forces A=B$.
\end{lem}
\begin{proof}
By the definition $p(m\mapsto \zer)\forces A$ and $p(m\mapsto \one)\forces A$ and $p(m\mapsto \zer)\forces B$ and $p(m\mapsto \one)\forces B$. By Lemma~\ref{LEM:typeIntroLocal} $p\forces A$ and $p\forces B$. By Lemma~\ref{LEM:localitysublemma1} $p\forces A=B$.
\end{proof}

\begin{lem}\leavevmode
\label{LEM:termIntroEqLocal}
\begin{enumerate}
\item If $p(m\mapsto \zer) \forces t\typ A$ and $p(m\mapsto \one) \forces t\typ A$ for some $m\notin \dom{p}$ then $p\forces t\typ A$.
\item If $p(m\mapsto \zer) \forces t=u\typ A$ and $p(m\mapsto \one) \forces t=u\typ A$ for some $m\notin \dom{p}$ then $p\forces t=u\typ A$.
\end{enumerate}
\end{lem}
\begin{proof}
We prove the two statements mutually by induction.
\begin{enumerate}
 \item \label{termintromontone} From $p(m\mapsto \zer)\forces t\typ A$ and $p(m\mapsto \one)\forces t\typ A$ we have $p(m\mapsto \zer)\forces A$ and $p(m\mapsto \one)\forces A$ and by Lemma~\ref{LEM:typeIntroLocal} $p\forces A$. By induction on the derivation of $p\forces A$.
\begin{enumerate}
	\item \label{natlocaltermforced} Let $p\forces A$ by $\lfor_\tnat$. Since $t$ has proper $p(m\mapsto \zer)$-whnf and $p(m\mapsto \one)$-whnf. By Lemma~\ref{properwhnflocal} $t$ has a proper $p$-whnf. By induction on the derivation of $p(m\mapsto \zer)\forces t\typ A$.
	\begin{enumerate}	
		\item Let $p(m\mapsto \zer)\vdash t\rtred \xoverline{n}\typ A$
		\begin{enumerate}
			\item If $t$ has a canonical $p$-whnf then $p\vdash t\rtred \xoverline{n}\typ A$ and $p\forces t\typ A$ directly.
			\item Otherwise, $p\vdash t\rtred \mathbb{E}[\f\,\xoverline{k}]\typ A,k\notin\dom{p}$. But then we have that $k=m$ and by the definition $p\forces t\typ A$. 
		\end{enumerate}
		\item Let $p(m\mapsto \zer)\vdash t\rtred \mathbb{E}[\f\,\xoverline{k}],k\notin\dom{p(m\mapsto \zer)}$ and $p(m\mapsto \zer)(k\mapsto \zer)\forces t\typ A$ and $p(m\mapsto \zer)(k\mapsto \one)\forces t\typ A$. By monotonicity $p(m\mapsto \one)(k\mapsto \zer)\forces t\typ A$ and $p(m\mapsto \one)(k\mapsto \one)\forces t\typ A$. By IH $p(k\mapsto \zer)\forces t\typ A$ and $p(k\mapsto \one)\forces t\typ A$ and by the definition $p\forces t\typ A$.
	\end{enumerate}	
	The statement follows similarly when $p\forces A$ by $\lfor_\tunit$ and $\lfor_\tbool$.
	
	\item \label{pilocaltermforced} Let $p\forces A$ by $\lfor_\Pi$ and $p\vdash A\rtred \Pi(x\typ F) G$. Let $q\leqslant p$. If $q\leqslant p(m\mapsto b)$  then we have directly $q\forces a \typ F\Rightarrow q\forces t\,a\typ G[a]$ and $q\forces a=b\typ F \Rightarrow q\forces t\,a = t\,b \typ G[a]$. Otherwise, we have $q(m\mapsto b)\leqslant p(m\mapsto b)$. Let $q\forces a\typ F$. By monotonicity $q(m\mapsto \zer)\forces a\typ F$ and $q(m\mapsto \one)\forces a\typ F$ and we have $q(m\mapsto \zer)\forces t\,a\typ G[a]$ and $q(m\mapsto \one)\forces t\,a \typ G[a]$. By IH we have $q\forces t\,a\typ G[a]$. Let $q\forces a=b\typ F$. By monotonicity $q(m\mapsto \zer)\forces a=b\typ F$ and $q(m\mapsto \one)\forces a=b\typ F$ and we have $q(m\mapsto \zer)\forces t\,a=t\,b\typ G[a]$ and $q(m\mapsto \one)\forces t\,a=t\,b\typ G[a]$. By IH (\ref{termeqlocal}) $q\forces t\,a =t\,b \typ G[a]$. Thus we have $p\forces t\typ A$.
	
	\item \label{sigmalocaltermforced} Let $p\forces A$ by $\lfor_\Sigma$ and let $p\vdash A\rtred \Sigma(x\typ F) G$. We have $p(m\mapsto \zer)\forces t\fpr \typ F$ and $p(m\mapsto \one)\forces t\fpr \typ F$ and $p(m\mapsto \zer)\forces t\spr \typ G[t\fpr]$ and $p(m\mapsto \one)\forces t\spr \typ G[t\fpr]$. By IH $p\forces t\fpr \typ F$ and $p\forces t\spr \typ G[t\fpr]$. Thus $p\forces t\typ A$.
	\item Let $p\forces A$ by $\lfor_\tuniv$. The statement then follows similarly to Lemma~\ref{LEM:typeIntroLocal}.
	\item \label{locallocaltermforced} Let $p\forces A$ by $\lfor_\mathrm{Loc}$ and let $p\vdash A\rtred \mathbb{E}[\f\,\xoverline{k}], k\notin\dom{p}$. If $k=m$ then by the definition $p\forces t\typ A$. If $k\neq m$ then by monotonicity $p(k\mapsto \zer)(m\mapsto \zer)\forces t\typ A$ and $p(k\mapsto \zer)(m\mapsto \one)\forces t\typ A$ and $p(k\mapsto \one)(m\mapsto \zer)\forces t\typ A$ and $p(k\mapsto \one)(m\mapsto \one)\forces t\typ A$. By IH $p(k\mapsto \zer)\forces t\typ A$ and $p(k\mapsto \one)\forces t\typ A$. By the definition $p\forces t\typ A$.
\end{enumerate}

\item \label{termeqlocal} From $p(m\mapsto \zer)\forces t=u \typ A$ and $p(m\mapsto \zer)\forces t=u\typ A$ we have $p(m\mapsto \zer)\forces t\typ A$ and $p(m\mapsto \one)\forces t\typ A$ and $p(m\mapsto \zer)\forces u\typ A$ and $p(m\mapsto \one)\forces u\typ A$ and $p(m\mapsto \zer)\forces A$ and $p(m\mapsto \one)\forces A$. By Lemma~\ref{LEM:typeIntroLocal} $p\forces A$. By induction on the derivation of $p\forces A$.
\begin{enumerate}
	\item Let $p\forces A$ by $\lfor_\tnat$. By (\ref{natlocaltermforced}) we have $p\forces t\typ A$. By induction on the derivation of $p\forces t\typ A$.
	\begin{enumerate}
		\item If $p\vdash t\rtred \xoverline{n}\typ A$. By induction on the derivation of $p\forces u\typ A$
		\begin{enumerate}
			\item If $u$ has a canonical $p$-whnf then $p\vdash u\rtred \xoverline{n}\typ A$ and $p\forces t=u\typ A$.
			\item \label{natermeqlocalright} Otherwise, $p\forces u\rtred \mathbb{E}[\f\,\xoverline{k}]\typ A, k\notin\dom{p}$. If $k=m$ then by the definition $p\forces t=u \typ A$. If $k\neq m$ then by monotonicity $p(k\mapsto \zer)(m\mapsto \zer)\forces t=u\typ A$ and $p(k\mapsto \zer)(m\mapsto \one)\forces t= u\typ A$ and $p(k\mapsto \one)(m\mapsto \zer)\forces t=u\typ A$ and $p(k\mapsto \one)(m\mapsto \one)\forces t= u\typ A$. By IH, $p(k\mapsto \zer)\forces t=u \typ A$ and $p(k\mapsto \one)\forces t=u\typ A$. By the definition $p\forces t=u\typ A$.
		\end{enumerate}
		\item If $p\vdash t\rtred \mathbb{E}[\f\,\xoverline{k}]\typ A, k\notin\dom{p}$ and $p(k\mapsto \zer)\forces t\typ A$ and $p(k\mapsto \one)\forces t\typ A$. The statement follows similarly to (\ref{natermeqlocalright}).
	\end{enumerate}
	The statement follows similarly when $p\forces A$ holds by $\lfor_\tunit$ and $\lfor_\tbool$.
	
	\item Let $p\forces A$ by $\lfor_\Pi$ and $p\vdash A\rtred \Pi(x\typ F) G$. By (\ref{pilocaltermforced}) we have $p\forces t\typ A$ and $p\forces u\typ A$. Let $q\leqslant p$. If $q\leqslant p(m\mapsto b)$ for some $b\in \{\zer,\one\}$ then we have $q\forces a \typ F\Rightarrow q\forces t\,a=u\,a\typ G[a]$. Otherwise, we have $q(m\mapsto b)\leqslant p(m\mapsto b)$. Let $q\forces a\typ F$. By monotonicity $q(m\mapsto \zer)\forces a\typ F$ and $q(m\mapsto \one)\forces a\typ F$ and we have $q(m\mapsto \zer)\forces t\,a=u\,a\typ G[a]$ and $q(m\mapsto \one)\forces t\,a=u\,a \typ G[a]$. By IH we have $q\forces t\,a=u\,a\typ G[a]$. Thus $p\forces t=u\typ A$.
	
	\item \label{sigmalocaltermforced} Let $p\forces A$ by $\lfor_\Sigma$ and let $p\vdash A\rtred \Sigma(x\typ F) G$. By (\ref{sigmalocaltermforced}) $p\forces t\typ A$ and $p\forces u\typ A$. We have $p(m\mapsto \zer)\forces t\fpr = u\fpr \typ F$ and $p(m\mapsto \one)\forces t\fpr=u\fpr \typ F$ and $p(m\mapsto \zer)\forces t\spr=u\spr \typ G[t\fpr]$ and $p(m\mapsto \one)\forces t\spr=u\spr \typ G[t\fpr]$. By IH $p\forces t\fpr=u\fpr \typ F$. Since we have $p(m\mapsto \zer)\forces u\spr \typ G[t\fpr]$ and $p(m\mapsto \one)\forces u\spr \typ G[t\fpr]$ then by IH(\ref{termintromontone}) $p\forces u\spr \typ G[t\fpr]$. By IH we have $p\forces t\spr = u\spr \typ G[t\fpr]$. Thus we have $p\forces t=u\typ A$.
	\item Let $p\forces A$ by $\lfor_\tuniv$. The statement then follows similarly to Lemma~\ref{LEM:typeEqLocal}.
	\item Let $p\forces A$ by $\lfor_\mathrm{Loc}$ and let $p\vdash A\rtred \mathbb{E}[\f\,\xoverline{k}], k\notin\dom{p}$. By (\ref{locallocaltermforced}) $p\forces t\typ A$ and $p\forces u\typ A$. If $k=m$ then by the definition $p\forces t=u\typ A$. If $k\neq m$ then by monotonicity $p(k\mapsto \zer)(m\mapsto \zer)\forces t=u\typ A$ and $p(k\mapsto \zer)(m\mapsto \one)\forces t=u\typ A$ and $p(k\mapsto \one)(m\mapsto \zer)\forces t=u\typ A$ and $p(k\mapsto \one)(m\mapsto \one)\forces t=u\typ A$. By IH $p(k\mapsto \zer)\forces t=u\typ A$ and $p(k\mapsto \one)\forces t=u\typ A$. By the definition $p\forces t=u\typ A$.\qedhere
\end{enumerate}
\end{enumerate}
\end{proof}

\begin{cor}[Locality]
	\label{LEM:localCharacter}
	If $p \covt S$ and $q \forces J$ for all $q\in S$ then $p\forces J$.
\end{cor}
\begin{proof}
Follows from Lemma~\ref{LEM:typeIntroLocal}, Lemma~\ref{LEM:termIntroEqLocal}, and Lemma~\ref{LEM:typeEqLocal} by induction.
\end{proof}

\begin{lem}
\label{LEM:equivalentTypesHaveEqualPers}
\renewcommand{\theenumi}{\alph{enumi}}
Let $p\forces A=B$. \begin{enumerate}
\item If $p \forces t \typ  A$ then $p\forces t \typ  B$ and if $p\forces u\typ B$ then $p\forces u\typ A$. 
\item If $p \forces t = u\typ  A$ then $p \forces t=u\typ  B$ and if $p\forces v=w \typ B$ then $p\forces v=w \typ A$.
\end{enumerate}
\end{lem}
\begin{proof}
By induction on the derivations of $p\forces A$, $p\forces B$ and $p\forces A=B$
\begin{enumerate}
\item Let $p\vdash A\rtred \tnat$ and $p\vdash B\rtred \tnat$
\begin{enumerate}
	\item Let $p\forces t\typ A$. By Lemma~\ref{LEM:BaseTypeCanonicity} there is a partition $p\covt S$ where for each $q\in S$ $q\vdash t\rtred \xoverline{n}\typ A$ for some $n\in \nats$. But then $q\vdash t\rtred \xoverline{n}\typ B$ and $q\forces t\typ B$. By locality $p\forces t\typ B$. Similarly if $p\forces u\typ B$ then $p\forces u\typ A$.
	\item Let $p\forces t=u\typ A$ then there is a partition $p\covt S$ where for each $q\in S$ $q\vdash t\rtred \xoverline{n}\typ A$ and $q\vdash u\rtred \xoverline{n}\typ A$ for some $n\in \nats$. But then $q\vdash t\rtred \xoverline{n}\typ B$ and $q\vdash u\rtred \xoverline{n}\typ B$ and $q\forces t=u\typ B$. By locality $p\forces t=u\typ B$. Similarly $p\forces v=w\typ A$ whenever $p\forces v=w\typ B$
\end{enumerate}
\item Let $p\vdash A\rtred \Pi(x\typ F) G$ and $p\vdash B\rtred \Pi(x\typ H) E$.
\begin{enumerate}
	\item Let $p\forces t\typ A$ and $q\leqslant p$. Let $q\forces a\typ H$. From $p\forces A=B$ we get $p\forces F=H$ and by monotonicity $q\forces F=H$. By IH $q\forces a\typ F$. Thus we have $q\forces t\,a\typ G[a]$. Since $q\forces G[a]=E[a]$, by IH $q\forces t\,a \typ E[a]$. Similarly if $q\forces a=b\typ H$ by monotonicity $q\forces a=b\typ H$ and by IH $q\forces a=b\typ F$. Thus $q\forces t\,a = t\,b \typ G[a]$ and since $q\forces G[a]=E[a]$. By IH $q\forces t\,a=t\,b \typ E[a]$. Similarly $p\forces u\typ A$ when $p\forces u\typ B$.
	\item Let $p\forces t=u\typ A$ and $q\leqslant p$. Let $q\forces a\typ H$. Similarly to the above we get $q\forces t\,a = u\,a \typ \typ E[a]$. Thus showing $q\forces t=u\typ B$. Similarly we have $q\forces v=w\typ A$ when $q\forces v=w\typ B$
\end{enumerate}
\item Let $p\vdash A\rtred \Sigma(x\typ F) G$ and $p\vdash B\rtred \Sigma(x\typ H) E$.
\begin{enumerate}
	\item Let $p\forces t\typ A$. We have $p\forces t\fpr \typ F$ and $p\forces t\spr \typ G[t\fpr]$. From $p\forces A=B$ we get $p\forces F=H$ and by IH $p\forces t\fpr \typ H$. From $p\forces A=B$ we get $p\forces G[t\fpr]=E[t\fpr]$ and by IH $p\forces t\spr \typ E[t\fpr]$. Thus $p\forces t\typ B$. Similarly we have $p\forces u\typ A$ when $p\forces u\typ B$.
	
	\item Let $p\forces t=u\typ A$. We have $p\forces t\fpr =u\fpr \typ F$ and $p\forces t\spr =u\spr\typ G[t\fpr]$. From $p\forces A=B$ we get $p\forces F=H$ and by IH $p\forces t\fpr=u\fpr \typ H$. From $p\forces A=B$ we get $p\forces G[t\fpr]=E[t\fpr]$ and by IH $p\forces t\spr=u\spr \typ E[t\fpr]$. Thus $p\forces t=u\typ B$. Similarly we have $p\forces v=w\typ A$ when $p\forces v=w\typ B$.
	\end{enumerate}
	\item If either $A$ or $B$ does not reduce to a canonical $p$-whnf then by Lemma~\ref{forcedtypeslocallyreducetocanonical} we have a partition $p\covt S$ where for each $q \in S$ both $A$ and $B$ have canonical whnf and we can show the statement for each $q\in S$ by the above. By locality the statement follows for~$p$.\qedhere
\end{enumerate}
\end{proof}

Immediately from the definition we have
\begin{lem}
If $p \forces A$ then $p\forces A=A$.
\end{lem}

\begin{lem}
If $p\forces A=B$ then $p\forces B=A$.
\end{lem}
\begin{proof}
If both $A$ and $B$ have canonical $p$-whnf then the statement follows by induction from the definition and Lemma~\ref{LEM:equivalentTypesHaveEqualPers}. Otherwise, by Lemma~\ref{forcedtypeslocallyreducetocanonical} we have a partition $p\covt S$ where both $A$ and $B$ have canonical $q$-whnf for all $q\in S$. By monotonicity $q\forces A=B$ and it follows by the above that $q\forces B=A$. By locality $p\forces B=A$.
\end{proof}

\begin{lem}
\label{LEM:typeequiv_refl,symm,trans}
If $p\forces A=B$ and  $p\forces B=C$ then $p \forces A=C$.
\end{lem}
\begin{proof}
\leavevmode
Let $p\forces A=B$ and $p\forces B=C$. We then have that $p\forces A$, $p\forces B$ and $p\forces C$. Thus $A$, $B$ and $C$ have proper $p$-whnf. If any of these proper $p$-wnf is not canonical then by Lemma~\ref{forcedtypeslocallyreducetocanonical} we can find a partition $p\covt S$ where all three have canonical $q$-whnf for all $q\in S$. By monotonicity $q\forces A=B$ and $q\forces B=C$ for all $q\in S$. If we can then show that $q\forces A=C$ for all $q\in S$ then by locality we will have $p\forces A=C$. Thus we can assume w.l.o.g that $A$, $B$ and $C$ have canonical $p$-whnf.

By induction on the derivations of $p\forces A$, 
\begin{enumerate}[leftmargin=10mm]
\item Let $p\forces A$ by $\lfor_\tnat$. Since by assumption $B$ has a canonical $p$-wnf and $p\forces A=B$ then $p\vdash B\rtred \tnat$. Similarly $p\vdash C\rtred \tnat$ and we have $p\forces A=C$

The statement follows similarly when $p\forces A$ holds by $\lfor_\tempt$, $\lfor_\tunit$ and $\lfor_\tbool$.

\item Let $p\forces A$ by $\lfor_\Pi$ and $p\vdash A\rtred \Pi(x\typ F) G$. From $p\forces A=B$ and since by assumption $B$ has a canonical $p$-whnf we have $p\vdash B\rtred \Pi(x\typ H) E$ and $p\forces F=H$ and $\forall q\leqslant p (q\forces a \typ F \Rightarrow q\forces G[a]=E[a])$. Since $p\forces B=C$ and by assumption $C$ has a canonical $p$-whnf we have $p\vdash C\rtred \Pi(x\typ T) R$ and $p\forces H=T$ and $\forall q\leqslant p (q\forces b\typ H \Rightarrow q\forces E[b]=R[b])$. 

By IH $p\forces F=T$. Let $q\leqslant p$ and $q\forces a\typ F$. By monotonicity $q\forces F=H$ and by Lemma~\ref{LEM:equivalentTypesHaveEqualPers} $q\forces a\typ H$. Thus $q\forces E[a]=R[a]$. But $q\forces G[a]=E[a]$. By IH $q\forces G[a]=R[a]$. Thus we have $p\forces A=C$.

The statement follows similarly when $p\forces A$ holds by $\lfor_\Sigma$. 

\item[$(\mathbf{F_\tuniv})$] Since $p\forces A=B$ and $p\forces B=C$, we have $B\coloneqq \tuniv$ and $C\coloneqq \tuniv$ and the statements follows.\qedhere
\end{enumerate}
\end{proof}

Immediately from the definition we have the following 
\begin{lem}
If $p\forces t\typ A$ then $p\forces t=t\typ A$.
\end{lem}

\begin{lem}
If $p\forces t=u\typ A$ and then $p \forces u=t\typ A$.
\end{lem}
\begin{proof}
Let $p\forces t=u\typ A$. We have $p\forces t\typ A$, $p\forces u\typ A$ and $p\forces A$. By induction on the derivation of $p\forces A$.
\begin{enumerate}
\item Let $p\forces A$ by $\lfor_\tnat$. Since $p\forces t=u\typ A$ we have a partition (Lemma~\ref{LEM:BaseTypeCanonicity}) $p\covt S$ where for each $q\in S$ we have $q\vdash t\rtred \xoverline{n}\typ A$ and $q\vdash u\rtred \xoverline{n}\typ A$ for some $n\in \nats$. Hence $q\forces u=t \typ A$ for all $q\in S$. By locality $p\forces t=u\typ A$.

The statement follows similarly when $p\forces A$ is derived by $\lfor_\tunit$ and $\lfor_\tbool$.

\item Let $p\forces A$ by $\lfor_\Pi$ and let $p\vdash A\rtred \Pi(x\typ F) G$. Let $q\leqslant p$ and $q\forces a\typ F$ we then have $q\forces t\,a=u\,a \typ G[a]$. We have $q\forces G[a]$ and by IH $q\forces u\,a=t\,a \typ G[a]$. Thus $p\forces u=t\typ A$.
  
\item Let $p\forces A$ by $\lfor_\Sigma$ and let $p\vdash A\rtred \Sigma(x\typ F) G$. We have $p\forces t\fpr=u\fpr \typ F$ and $p\forces t\spr = u\spr \typ G[t\fpr]$. Since $q\forces F$, by IH $p\forces u\fpr = t\fpr \typ F$. Since $p\forces A$ we have  $p\forces G[t\fpr]=G[u\fpr]$. By Lemma~\ref{LEM:equivalentTypesHaveEqualPers} $p\forces t\spr = u\spr \typ G[u\fpr]$. Since $p\forces G[u\fpr]$, by IH $p\forces u\spr = t\spr \typ G[u\fpr]$. Thus $p\forces u=t\typ A$.

\item Let $p\forces A$ by $\lfor_\tuniv$. The statement then follows similarly to Lemma~\ref{LEM:typeequiv_refl,symm,trans}

\item  Let $p\forces A$ by $\lfor_\mathrm{Loc}$ and let $p\vdash A\rtred \mathbb{E}[\f\,\xoverline{k}], k\notin\dom{p}$ and $p(k\mapsto \zer)\forces A$ and $p(k\mapsto \one)\forces A$. Since $p\forces t=u\typ A$ we have that $p(k\mapsto \zer)\forces t=u\typ A$ and $p(k\mapsto \one)\forces t=u \typ A$. By IH $p(k\mapsto \zer)\forces u=t\typ A$ and $p(k\mapsto \one)\forces u=t\typ A$. By the definition $p\forces u=t\typ A$.\qedhere
\end{enumerate}
\end{proof}

\begin{lem}
If $p\forces t=u\typ A$ and $p\forces u=v\typ A$ then $p\forces t=v\typ A$.
\end{lem}
\begin{proof}
Let $p\forces t=u\typ A$ and  $p\forces u=v\typ A$. We have $p\forces A$, $p\forces t\typ A$, $p\forces u\typ A$ and $p\forces v\typ A$. By induction on the derivation of $p\forces A$.
\begin{enumerate}
\item Let $p\forces A$ by $\lfor_\tnat$. By Lemma~\ref{LEM:BaseTypeCanonicity} there is a partition $p\covt S$ where for each $q\in S$ we have $q\vdash t\rtred \xoverline{n}\typ A$, $q\vdash u\rtred \xoverline{n}\typ A$, and $q\vdash v\rtred \xoverline{n}\typ A$. Thus $q\forces t=v\typ A$ for all $q\in S$. By locality $q\forces t=v\typ A$.

The statement follows similarly when $p\forces A$ by $\lfor_\tunit$ and $\lfor_\tbool$

\item Let $p\forces A$ by $\lfor_\Pi$ and let $p\vdash A\rtred \Pi(x\typ F) G$. Let $q\leqslant p$ and $q\forces a\typ F$. We have then $q\forces t\,a=u\,a \typ G[a]$ and $q\forces u\,a=v\,a\typ G[a]$. By IH $q\forces t\,a = v\,a \typ G[a]$. Thus $p\forces t=v\typ A$.

\item Let $p\forces A$ by $\lfor_\Sigma$ and let $p\vdash A\rtred \Sigma(x\typ F) G$. Since $p\forces t=u\typ A$ we have that $p\forces t\fpr = u\fpr \typ F$ and $p\forces t\spr = u\spr \typ G[t\fpr]$. Similarly we have that $p\forces u\fpr = v\fpr \typ F$ and $p\forces u\spr = v\spr \typ G[u\fpr]$. Since $p\forces A$ we have that $p\forces G[t\fpr]=G[u\fpr]$ and by Lemma~\ref{LEM:equivalentTypesHaveEqualPers} $p\forces u\spr = v\spr \typ G[t\fpr]$. By IH we have $p\forces t\fpr = v\fpr \typ F$ and $p\forces t\spr = v\spr \typ G[t\fpr]$. We have then that $p\forces t= v\typ A$.

\item Let $p\forces A$ by $\lfor_\tuniv$. The statement then follows similarly to Lemma~\ref{LEM:typeequiv_refl,symm,trans}.

\item Let $p\forces A$ by $\lfor_\mathrm{Loc}$ and let $p\vdash A\rtred \mathbb{E}[\f\,\xoverline{k}], k\notin\dom{p}$ and $p(k\mapsto \zer)\forces A$ and $p(k\mapsto \one)\forces A$. We have that $p(k\mapsto \zer)\forces t=u\typ A$ and $p(k\mapsto \one)\forces t=u \typ A$ and $p(k\mapsto \zer)\forces u=v\typ A$ and $p(k\mapsto \one)\forces u=v\typ A$. By IH $p(k\mapsto \zer) \forces t=v\typ A$ and $p(k\mapsto \one)\forces t=v\typ A$. By the definition $p\forces t=v\typ A$.\qedhere
\end{enumerate}
\end{proof}

\noindent
Theorem ~\ref{properties} then follows from the above.

\section{Soundness}
In this section we show that the type theory described in Section~\ref{SEC:Forcingextensionoftypetheory} is sound with respect to the semantics described in Section~\ref{SEC:semantics}. I.e.\ we aim to show that $p\forces J$ whenever $\vdash_p J$. 

\begin{lem}
\label{LEM:types_reductionimplyequiv}
If $p\vdash A \rtred B$ and $p\forces B$ then $p\forces A$ and $p\forces A=B$.
\end{lem}
\begin{proof}
Follows from the definition by induction on the derivation of $p\forces B$.
\end{proof}
\begin{lem}
\label{LEM:terms_onestepreductionimplyequiv}
Let $p\forces A$. If $p\vdash t \red u\typ A$ and $p\forces u\typ A$ then $p\forces t\typ A$ and $p\forces t=u\typ A$.
\end{lem}
\begin{proof}
Let $p\vdash t \red u\typ A$ and $p\forces u\typ A$. By induction on the derivation of $p\forces A$.
\begin{enumerate}
\item Let $p\forces A$ by $\lfor_\tuniv$. The statement follows similarly to Lemma~\ref{LEM:types_reductionimplyequiv}. 

\item Let $p\forces A$ by $\lfor_\tnat$. By induction on the derivation of $p\forces u\typ A$. If $p\vdash u\rtred \xoverline{n} \typ \tnat$ then $p\vdash t\rtred \xoverline{n} \typ \tnat$ and the statement follows by the definition. If $p\vdash u\rtred \mathbb{E}[\f\,\xoverline{k}]\typ A, k\notin \dom{p}$ and $p(k\mapsto \zer)\forces u\typ A$ and $p(k\mapsto \one)\forces u\typ A$ then since $p(k\mapsto b)\vdash t\red u\typ A$, by IH  $p(k\mapsto \zer)\forces t\typ A$ and $p(k\mapsto \one)\forces t\typ A$ and $p(k\mapsto \zer)\forces t=u\typ A$ and $p(k\mapsto \one)\forces t= u\typ A$. By the definition $p\forces t\typ A$ and $p\forces t=u \typ A$. 

The statement follows similarly for $\lfor_\tunit$ and $\lfor_\tbool$.

\item Let $p\forces A$ by $\lfor_\Pi$ and let $p\vdash A \rtred \Pi(x\typ F)G$. Let $q\leqslant p$ and $q\forces a\typ F$. We have $q\vdash t\,a\red u\,a \typ G[a]$. By IH $q\forces t\,a \typ G[a]$ and $q\forces t\,a=u\,a \typ G[a]$. If $q\forces a=b \typ F$ we similarly get $q\forces t\,b\typ G[b]$ and $q\forces t\,b=u\,b\typ G[b]$. Since $q\forces G[a]=G[b]$, by Lemma~\ref{LEM:equivalentTypesHaveEqualPers}  $q\forces t\,b=u\,b\typ G[a]$. But $q\forces u\,a=u\,b\typ G[a]$. By symmetry and transitivity $q\forces t\,a = t\,b \typ G[a]$. Thus $p\forces t\typ A$ and $p\forces t=u\typ A$.

\item Let $p\forces A$ by $\lfor_\Sigma$ and let $p\vdash A \rtred \Sigma(x\typ F)G$. From $p\vdash t\red u\typ A$ we have $\vdash_p t\typ A$ and we have $p\vdash t\fpr \red u\fpr\typ F$ and $p\vdash t\spr\red u\spr\typ G[u\fpr]$. By IH $p\forces t\fpr \typ F$ and $p\forces t\fpr = u\fpr \typ F$. By IH $p\forces t\spr \typ G[u\fpr]$ and $p\forces t\spr = u\spr \typ G[u\fpr]$. But since $p\forces A$ and we have shown $p\forces t\fpr = u\fpr \typ F$ we get $p\forces G[t\fpr]=G[u\fpr]$. By Lemma~\ref{LEM:equivalentTypesHaveEqualPers} $p\forces t\spr \typ G[t\fpr]$ and $p\forces t\spr=u\spr\typ G[t\fpr]$. Thus $p \forces t\typ A$ and $p \forces t=u\typ A$

\item Let $p\forces A$ by $\lfor_\mathrm{Loc}$. Let $p\vdash A\rtred \mathbb{E}[\f\,\xoverline{k}], k\notin\dom{p}$ and $p(k\mapsto \zer)\forces A$ and $p(k\mapsto \one)\forces A$. Since $p\forces u\typ A$ we have $p(k\mapsto \zer)\forces u\typ A$ and $p(k\mapsto \one)\forces u\typ A$. But we have $p(k\mapsto b)\vdash t\red u\typ A$. By IH $p(k\mapsto \zer)\forces t\typ A$ and $p(k\mapsto \one)\forces t\typ A$ and $p(k\mapsto \zer)\forces t=u\typ A$ and $p(k\mapsto \one)\forces t= u \typ A$. By the definition $p\forces t\typ A$ and $p\forces t=u \typ A$. \qedhere
\end{enumerate}
\end{proof}

\begin{cor}
\label{COR:terms_reductionimplyequiv}
Let $p\vdash t \rtred u\typ A$ and $p\forces A$. If $p\forces u\typ A$ then $p\forces t\typ A$ and $p\forces t=u\typ A$.
\end{cor}

\begin{cor}
\label{COR:genericwelltyped}
$\forces \f\typ  \tnat\rightarrow \tbool$. 
\end{cor}
\begin{proof}
It's direct to see that $\forces \tnat\rightarrow \tbool$. For an arbitrary condition $p$ let $p\forces n\typ \tnat$. By Lemma~\ref{LEM:BaseTypeCanonicity} we have a partition $p\covt S$ where for each $q\in S$, $q\vdash n\rtred \xoverline{m}\typ \tnat$ for some $m\in\nats$.  We have thus a reduction $q\vdash \f\,n\rtred \f\,\xoverline{m}\typ \tbool$. If $m\in \dom{q}$ then $q\vdash \f\,n\rtred \f\,\xoverline{m}\red q(m) \typ \tbool$ and by definition $q \forces \f\,n \typ \tbool$. If $m \notin \dom{q}$ then $q(m\mapsto \zer) \vdash \f\,n\rtred \f\xoverline{m}\red \zer\typ \tbool$ and $q(m_j\mapsto \one) \vdash \f\,n\rtred \f\xoverline{m}\red \one\typ \tbool$. Thus $q(m\mapsto \zer) \forces \f\,n \typ \tbool$ and $q(m\mapsto \one)\forces \f\,n \typ \tbool$. By the definition $q \forces \f\,n \typ \tbool$. We thus have that $q \forces \f\,n \typ \tbool$ for all $q\in S$ and by locality $p\forces \f\,n \typ \tbool$. 

Let $p\forces a=b\typ \tnat$. By Lemma~\ref{LEM:BaseTypeCanonicity} there is a partition $p\covt S$ where for each $q\in S$, $q\vdash a\rtred \xoverline{m} \typ \tnat$ and $q\vdash b \rtred \xoverline{m}\typ \tnat$ for some $m\in \nats$. We then have $q \vdash \f\,a \rtred \f\,\xoverline{m} \typ \tbool$ and $q\vdash \f\,b\rtred \f\,\xoverline{m}\typ \tbool$. If $m\in \dom{q}$ then $q \vdash \f\,a\rtred q(m)\typ \tbool$ and $q\vdash \f\,b\rtred q(m)\typ \tbool$. By Corollary~\ref{COR:terms_reductionimplyequiv}, symmetry and transitivity $q \forces \f\,a=\f\,b \typ \tbool$. If on the other hand $m \notin\dom{q}$ then $q(m\mapsto \zer) \forces \f\,a=\f\,b \typ \tbool$ and $q(m\mapsto \one)\forces \f\,a=\f\,b \typ \tbool$. By the definition $q \forces \f\,a=\f\,b\typ \tbool$. Thus $q \forces \f\,a=\f\,b\typ \tbool$ for all $q\in S$. By locality $p\forces \f\,a=\f\,b\typ \tbool$. Hence $\forces \f\typ \tnat\rightarrow \tbool$.
\end{proof}

\begin{lem}
\label{LEM:forcingTermNegation}
If $\vdash_p t\typ \neg A$ and $p\forces A$ then $p\forces t\typ \neg A$ iff for all $q\leqslant p$ there is no term $u$ such that $q\forces u\typ A$.
\end{lem}
\begin{proof}
Let $p\forces A$ and $\vdash_p t\typ \neg A$. We have directly that $p\forces \neg A$. Assume $p\forces t\typ \neg A$. If $q \forces u\typ A$ for some $q\leqslant p$, then $q\forces t\,u \typ \tempt$ which is impossible. Conversely, assume it is the case that for all  $q \leqslant p$ there is no $u$ for which $q\forces u\typ A$. Since $r\forces a\typ A$ and $r\forces a=b\typ A$ never hold for any $r\leqslant p$, the statements  ``$r\forces a\typ A \Rightarrow r\forces t\,a \typ \tempt$'' and ``$r\forces a=b \typ A \Rightarrow r\forces t\,a=t\,b \typ \tempt$'' hold trivially.
\end{proof}

\begin{lem}
\label{LEM:markovpremisewitnesswelltyped}
$\forces \w \typ \neg\neg (\Sigma(x\typ \tnat) \iszero(\f\,x))$.
\end{lem}
\begin{proof}
By Lemma~\ref{LEM:forcingTermNegation} it is enough to show that for all $q$ there is no term $u$ for which $q\forces u\typ \neg (\Sigma(x\typ \tnat) \iszero(\f\,x))$. Assume $q\forces u\typ\neg (\Sigma(x\typ \tnat) \iszero(\f\,x))$ for some $u$. Let $m\notin\dom{q}$ we have then $q(m\mapsto \zer) \forces (\overline{m},\zer) \typ \Sigma(x\typ \tnat) \iszero(\f\,x)$ thus $q(m\mapsto \zer) \forces u\,(\xoverline{m},\zer)\typ \tempt$ which is impossible.\end{proof}

Let $\Gamma \coloneqq x_1\typ A_1\dots,x_n\typ A_n$ and $\rho \coloneqq a_1,\dots,a_n$. Let $A_i\rho \coloneqq A_i[a_1/x_1,\dots,a_{i-1}/x_{i-1}]$. We write $\Delta \vdash_p \rho \typ \Gamma$ if $\Delta \vdash_p a_i\typ A_i\rho$ for all $i$. Letting $\sigma=b_1,\dots,b_n$, we write $\Delta \vdash_p \rho =\sigma\typ \Gamma$ if $\Delta \vdash_p a_i=b_i\typ A_i\rho$ for all $i$. We write $p \forces \rho\typ \Gamma$ if $p\forces a_i\typ A_i\rho$ for all $i$ and $p\forces \rho=\sigma\typ \Gamma$ if $p\forces a_i=b_i\typ A_i\rho$ for all $i$.

\begin{lem}
\label{LEM:evalissub}
If $p\forces \rho \typ \Gamma$ then $\vdash_p \rho \typ \Gamma$. If $p\forces \rho = \sigma \typ \Gamma$ then $\vdash_p \rho = \sigma \typ \Gamma$.
\end{lem}
\begin{proof}
Follows from the definition.
\end{proof}

\begin{defi}\leavevmode
\begin{enumerate}
\item We write $\Gamma\vDash_p A$ if $\Gamma\vdash_p A$ and for all $q\leqslant p$ whenever $q\forces \rho \typ \Gamma$ then $q\forces A\rho$ and whenever $q\forces \rho=\sigma\typ\Gamma$ then $q\forces A\rho=A\sigma$. 
\item We write $\Gamma\vDash_p t\typ A$ if $\Gamma\vdash_p t\typ A$, $\Gamma\vDash_p A$ and for all $q\leqslant p$ whenever $q\forces \rho \typ \Gamma$ then $q\forces t\rho\typ A\rho$ and whenever $q\forces \rho = \sigma \typ \Gamma$ then $q\forces t\rho=t\sigma\typ A\rho$.
\item We write $\Gamma\vDash_p A=B$ if $\Gamma\vdash_p A=B$, $\Gamma\vDash_p A$, $\Gamma\vDash_p B$ and for all $q\leqslant p$ whenever $q\forces \rho \typ \Gamma$ then $q\forces A\rho = B\rho$. 
\item We write $\Gamma\vDash_p t=u\typ A$ if $\Gamma\vdash_p t=u\typ A$, $\Gamma\vDash_p t\typ A$, $\Gamma\vDash_p u\typ A$  and for all $q\leqslant p$ whenever $q\forces \rho \typ \Gamma$ then $q\forces t\rho = u\rho \typ A\rho$.
\end{enumerate}
\end{defi}

In the following we will show that whenever we have a rule 
\AxiomC{$\Gamma_1\vdash_{p_1} J_1$} 
\AxiomC{$\dots$} 
\AxiomC{$\Gamma_\ell \vdash_{p_\ell} J_\ell$} 
\TrinaryInfC{$\Gamma\vdash_p J$}
\DisplayProof
 in the type system then it holds that 
\AxiomC{$\Gamma_1\vDash_{p_1} J_1$} 
\AxiomC{$\dots$} 
\AxiomC{$\Gamma_\ell \vDash_{p_\ell} J_\ell$} 
\TrinaryInfC{$\Gamma\vDash_p J$}
\DisplayProof.
 Which is sufficient to show soundness.

\begin{lem}
\AxiomC{$\Gamma \vDash_{p_1} J$}
\AxiomC{$\dots$}
\AxiomC{$\Gamma \vDash_{p_n} J $} 
\def\labelSpacing{1pt}
\RightLabel{$p \covt \{p_1,\dots, p_n\}$}
\TrinaryInfC{$\Gamma \vDash_p J $}
\DisplayProof\smallskip
\end{lem}
\begin{proof}
Follows from Corollary~\ref{LEM:localCharacter}.
\end{proof}

\begin{lem}\leavevmode \medskip
\label{soundPiIntro}
\AxiomC{$\Gamma \vDash_p F$ }
\AxiomC{$\Gamma, x\typ  F \vDash_p G$}
\BinaryInfC{$\Gamma \vDash_p \Pi (x\typ  F) G$}
\DisplayProof
\,
\AxiomC{$\Gamma \vDash_p F$ }
\AxiomC{$\Gamma, x\typ  F \vDash_p G$}
\BinaryInfC{$\Gamma \vDash_p \Sigma (x\typ  F) G$}
\DisplayProof
\end{lem}
\begin{proof}
Let $q\leqslant p$ and $q\forces \rho \typ \Gamma$. Let $r\leqslant q$. If $r\forces a\typ F\rho$, we have $r\forces (\rho,a)\typ (\Gamma,x\typ F)$, thus $r\forces G\rho[a]$. If moreover $r\forces a=b \typ F\rho$ then $r\forces (\rho,a)=(\rho,b)\typ(\Gamma,x\typ F)$ and we have $r\forces G\rho[a]=G\rho[b]$. Thus $q\forces (\Pi(x\typ F) G)\rho$ and $q\forces (\Sigma(x\typ F)G)\rho$

Let $q\forces \rho = \sigma \typ \Gamma$. We have that $r\forces F\rho=F\sigma$. If $r\forces a\typ F\rho$ then by Lemma~\ref{LEM:equivalentTypesHaveEqualPers} $r\forces a\typ F\sigma$. Thus $r\forces (\rho,a)=(\sigma,a)\typ (\Gamma,x\typ F)$. We have $r\forces G\rho[a]=G\sigma[a]$. Thus $q\forces (\Pi(x\typ F) G)\rho= (\Pi(x\typ F) G)\sigma$ and $q\forces (\Sigma(x\typ F)G)\rho = (\Sigma(x\typ F)G)\sigma$.\smallskip
\end{proof}

\begin{lem}
\label{soundpieq}

\AxiomC{$\Gamma\vDash_p F= H$}
\AxiomC{$\Gamma, x\typ  F \vDash_p G =E $}
\BinaryInfC{$\Gamma \vDash_p \Pi (x\typ F) G =\Pi (x\typ H) E$}
\DisplayProof
\,
\AxiomC{$\Gamma\vDash_p F= H$}
\AxiomC{$\Gamma, x\typ  F \vDash_p G =E $}
\BinaryInfC{$\Gamma \vDash_p \Sigma (x\typ F) G =\Sigma (x\typ H) E$}
\DisplayProof
\end{lem}
\begin{proof}

Let $q\leqslant p$ and $q\forces \rho \typ \Gamma$. Similarly to Lemma~\ref{soundPiIntro}, we can show $q\forces (\Sigma (x\typ F) G)\rho$, $q\forces (\Sigma (x\typ H) E)\rho$, $q\forces ( \Pi (x\typ F) G)\rho$, and $q\forces (\Pi (x\typ H) E)\rho$. 

From  $q\forces F\rho = H\rho $. Let  $r\leqslant q$ and $r\forces a\typ F\rho$. We have then $r\forces (\rho,a)\typ(\Gamma,x\typ F)$. Thus $r\forces G\rho[a]=E\rho[a]$. Thus $q\forces (\Pi(x\typ F)G)\rho=(\Pi(x\typ H)E)\rho$ and $q\forces (\Sigma(x\typ F) G)\rho = (\Sigma(x\typ H)E)\rho$.
\end{proof}

\begin{lem}
\label{soundLambdaIntro}
\AxiomC{$\Gamma, x \typ  F \vDash_p t \typ  G$ }
\UnaryInfC{$\Gamma \vDash_p \lambda x.t \typ  \Pi (x\typ F) G$}
\DisplayProof
\end{lem}
\begin{proof}
Let $q\leqslant p$ and $q\forces \rho \typ \Gamma$. By Lemma~\ref{LEM:evalissub} $\vdash_q \rho \typ \Gamma$. Let $r\leqslant q$ and $r\forces d\typ F\rho$. Since $\Gamma, x \typ  F \vdash_p t \typ  G$ we have that $x\typ F\rho \vdash_r t\rho \typ G\rho$. Since $\vdash_r d\typ F\rho$ we have $\vdash_r (\lambda x. t\rho) d = t\rho[d]\typ G\rho[d]$. By the reduction rules $(\lambda x. t\rho)d \red t\rho[d]$. Thus $r\vdash (\lambda x. t\rho) d \red t\rho[d]\typ G\rho[d]$. But $r\forces (\rho,d)\typ (\Gamma,x\typ F)$, hence, $r\forces t\rho[d]\typ G\rho[d]$. By Lemma~\ref{LEM:terms_onestepreductionimplyequiv} we have that $r\forces (\lambda x. t\rho)\,d \typ G\rho[d]$ and $r\forces (\lambda x. t\rho)\,d = t\rho[d]\typ G\rho[d]$. 

Let $r\forces e=d\typ F\rho$ we have similarly that $r\forces (\lambda x. t\rho)\,e = t\rho[e]\typ G\rho[e]$. We have also that $r\forces (\rho,d)=(\rho,e)\typ G\rho[d]$, thus  $r\forces t\rho[d]=t\rho[e]\typ G\rho[d]$ and $r\forces G\rho[d]=G\rho[e]$. By Lemma~\ref{LEM:equivalentTypesHaveEqualPers} we have $r\forces (\lambda x. t\rho)\,e = t\rho[e]\typ G\rho[d]$. By symmetry and transitivity we have $r\forces (\lambda x. t\rho)\,d=(\lambda x. t\rho)\,e\typ G\rho[d]$. Thus $q\forces (\lambda x.t)\rho \typ (\Pi(x\typ F) G)\rho$.

Let $q\forces \rho=\sigma\typ \Gamma$. We get $q\forces F\rho = F\sigma$.  Similarly to the above we can show  $q\forces (\lambda x.t)\sigma \typ (\Pi(x\typ F) G)\sigma$. Let $r\leqslant q$ and $r\forces a\typ F\rho$. By Lemma~\ref{LEM:equivalentTypesHaveEqualPers} $r\forces a\typ F\sigma$. We then have $r\forces (\rho,a)=(\sigma,a)\typ (\Gamma,x\typ F)$. Thus we have $r\forces G\rho[a]=G\sigma[a]$. Thus $q\forces  (\Pi(x\typ F) G)\rho =  (\Pi(x\typ F) G)\sigma$ and by Lemma~\ref{LEM:equivalentTypesHaveEqualPers} $q\forces (\lambda x.t)\sigma \typ (\Pi(x\typ F) G)\rho$. We have $r\forces t\rho[a]=t\sigma[a]\typ G\rho[a]$. But $r\forces (\lambda x. t\rho)\,a = t\rho[a]\typ G\rho[a]$ and $r\forces (\lambda x. t\sigma) a = t\sigma[a] \typ G\sigma[a]$. By Lemma~\ref{LEM:equivalentTypesHaveEqualPers} $r\forces (\lambda x. t\sigma) a = t\sigma[a] \typ G\rho[a]$. By Symmetry and transitivity $r\forces (\lambda x. t\rho)\, a = (\lambda x. t\sigma)\,a \typ G\rho[a]$. Thus $q\forces (\lambda x. t) \rho = (\lambda x. t) \sigma \typ (\Pi(x\typ F) G)\rho$.
\end{proof}

\begin{lem}
\label{soundbeta}

\AxiomC{$\Gamma, x\typ F \vDash_p t \typ  G$}
\AxiomC{$\Gamma \vDash_p a \typ  F$}
\BinaryInfC{$\Gamma \vDash_p (\lambda x. t)\,a = t[a] \typ  G[a]$}
\DisplayProof
\end{lem}
\begin{proof}
Let $q\leqslant p$ and $q\forces \rho \typ \Gamma$. We have $q\forces a\rho \typ F\rho$. As in Lemma~\ref{soundLambdaIntro} $q\vdash ((\lambda x. t)\,a)\rho \red t[a]\rho \typ G[a]\rho$ which by Lemma~\ref{LEM:terms_onestepreductionimplyequiv} imply that $q\forces ((\lambda x. t)\,a)\rho = t[a]\rho \typ G\rho[a]$.
\end{proof}

\begin{lem}
\label{soundapptyp}

\AxiomC{$\Gamma\vDash_p g \typ  \Pi (x\typ F) G$}
\AxiomC{$\Gamma\vDash_p a \typ  F$}
\BinaryInfC{$\Gamma \vDash_p g\,a \typ  G[a]$}
\DisplayProof
\end{lem}
\begin{proof}
Let $q\leqslant p$ and $q\forces \rho \typ \Gamma$. We have $q\forces g\rho \typ (\Pi(x\typ F) G)\rho$ and  $q\forces a\rho \typ F\rho$. By the definition $q\forces (g\,a)\rho \typ G[a]\rho$.

Let $q\forces \rho =\sigma\typ \Gamma$. We have then $q\forces g\rho = g\sigma \typ (\Pi(x\typ F) G)\rho$ and $q\forces a\rho = a\sigma \typ F\rho$. From the definition $q\forces g\rho\,a\rho = g\sigma\,a\rho \typ G[a]\rho$. From the definition $q\forces g\sigma\,a\rho = g\sigma\,a\sigma \typ G[a]\rho$. By transitivity $q\forces (g\,a)\rho = (g\,a)\sigma \typ G[a]\rho$. 
\end{proof}

\begin{lem}
\label{soundappeq}
(1) 
\AxiomC{$\Gamma \vDash_p g \typ  \Pi (x\typ F) G$}
\AxiomC{$\Gamma \vDash_p u =v \typ  F$}
\BinaryInfC{$\Gamma \vDash_p g\,u = g\,v \typ  G[u]$}
\DisplayProof\;
(2)
\AxiomC{$\Gamma \vDash_p h =g \typ  \Pi (x\typ F) G$}
\AxiomC{$\Gamma \vDash_p u \typ  F$}
\BinaryInfC{$\Gamma \vDash_p h\,u = g\,u \typ  G[u]$}
\DisplayProof 
\end{lem}
\begin{proof}
Let $q\leqslant p$ and $q\forces \rho \typ \Gamma$.  

\begin{enumerate}
\item We have $q\forces g\rho \typ (\Pi(x\typ F) G)\rho$ and  $q\forces u\rho = v\rho \typ F\rho$. From the definition get $q\forces (g\,u)\rho = (g\,v)\rho \typ G[u]\rho$.
\item We have  $q\forces h\rho = g\rho \typ (\Pi(x\typ F) G)\rho$ and  $q\forces u\rho \typ F\rho$. From the definition we get $q\forces (h\,u)\rho = (g\,u)\rho \typ G[u]\rho$.\qedhere
\end{enumerate}
\end{proof}

\begin{lem}
\label{soundfunext}

\AxiomC{$\Gamma \vDash_p h \typ \Pi(x\typ  F)G$}
\AxiomC{$\Gamma \vDash_p g \typ \Pi(x\typ  F) G$}
\AxiomC{$\Gamma, x\typ F \vDash_p h\,x=g\,x \typ G[x]$}
\TrinaryInfC{$\Gamma \vDash_p h = g \typ  \Pi(x\typ  F)G$}
\DisplayProof \smallskip
\end{lem}
\begin{proof}
Let $q\leqslant p$ and $q\forces \rho \typ \Gamma$. We have $q\forces h\rho\typ (\Pi(x\typ F) G)\rho$ and $q\forces g\rho\typ (\Pi(x\typ F) G)\rho$. Let $r\leqslant q$ and $r\forces a\typ F\rho$. We have then that $r\forces (\rho,a)\typ \Gamma,x\typ F$. Thus $r\forces h\rho\,a=g\rho\,a \typ G\rho[a]$. By the definition $q\forces h\rho = g\rho \typ (\Pi(x\typ F) G)\rho$.
\end{proof}

\begin{lem}
\label{soundpairing}

\AxiomC{$\Gamma, x\typ F \vDash_p G$ }
\AxiomC{$\Gamma \vDash_p a \typ  F$ }
\AxiomC{$\Gamma \vDash_p b\typ  G[a]$}
\TrinaryInfC{$\Gamma \vDash_p (a, b) \typ  \Sigma (x\typ  F) G $}
\DisplayProof
\end{lem}
\begin{proof}
Let $q\leqslant p$ and $q\forces \rho \typ \Gamma$. By the typing rules $\Gamma\vdash_q (a,b)\fpr = a \typ F$ and $\Gamma\vdash_q (a,b)\spr = b\typ F[a]$. But $\vdash_q \rho \typ \Gamma$. By substitution we have $\vdash_q ((a,b)\fpr)\rho = a\rho \typ F\rho$ and $\vdash_q ((a,b)\spr)\rho = b\rho \typ G[a]\rho$. But $((a,b)\fpr)\rho \red_q a\rho$ and $((a,b)\spr)\rho \red_q b\rho$. Thus $q\vdash ((a,b)\fpr)\rho \red  a\rho \typ F\rho$ and $q\vdash ((a,b)\spr)\rho \red  b\rho\typ G[a]\rho$. From the premise $q\forces a\rho \typ F\rho$ and $q\forces b\rho \typ G[a]\rho$. By Lemma~\ref{LEM:terms_onestepreductionimplyequiv}  $q\forces ((a,b)\fpr)\rho \typ F\rho$ and $q\forces ((a,b)\spr)\rho \typ G[a]\rho$. By Lemma~\ref{LEM:terms_onestepreductionimplyequiv} $q\forces ((a,b)\fpr)\rho = a\rho \typ F\rho$, thus $q\forces (\rho,a\rho)=(\rho,((a,b)\fpr)\rho)\typ (\Gamma,x\typ F)$. Hence $q\forces G[a]\rho = G[(a,b)\fpr]\rho$. By Lemma~\ref{LEM:equivalentTypesHaveEqualPers} $q\forces ((a,b)\spr)\rho \typ G[(a,b)\fpr]\rho$. By the definition we have then that $q\forces (a,b)\rho \typ (\Sigma(x\typ F) G)\rho$.

Let $q\forces \rho = \sigma \typ \Gamma$. Similarly we can show  $q\forces (a,b)\sigma \typ (\Sigma(x\typ F) G)\sigma$. We have that $q\forces a\rho = a\sigma \typ F\rho$ and $q\forces b\rho = b\sigma \typ G[a]\rho$. We have also $q\forces (\rho,a\rho)=(\sigma,a\sigma) \typ (\Gamma,x\typ F)$ we thus have $q\forces G[a]\rho = G[a]\sigma$. By Lemma~\ref{LEM:terms_onestepreductionimplyequiv} $q\forces ((a,b)\spr)\sigma = b\sigma \typ G[a]\sigma$. By Lemma~\ref{LEM:equivalentTypesHaveEqualPers} $q\forces ((a,b)\spr)\sigma = b\sigma \typ G[a]\rho$. But we also have by Lemma~\ref{LEM:terms_onestepreductionimplyequiv} that $q\forces ((a,b)\spr)\sigma = a\sigma \typ F\sigma$. Hence, by Lemma~\ref{LEM:equivalentTypesHaveEqualPers}, we have $q\forces ((a,b)\spr)\sigma=a\sigma\typ F\rho$. By symmetry and transitivity we then have that $q\forces ((a,b)\fpr)\rho = ((a,b)\fpr)\sigma \typ F\rho$ and $q\forces ((a,b)\spr)\rho = ((a,b)\spr)\sigma \typ G[(a,b)\fpr]\rho$. Thus we have that $q\forces (a,b)\rho = (a,b)\sigma \typ (\Sigma(x\typ F) G)\rho$.
\end{proof}

\begin{lem}
\label{soundproj} \leavevmode\\[2mm]
{\upshape (1)}
\AxiomC{$\Gamma, x\typ F\vDash_p G$}
\AxiomC{$\Gamma\vDash_p t\typ F$}
\AxiomC{$\Gamma\vDash_p u\typ G[t]$}
\TrinaryInfC{$\Gamma \vDash_p (t,u)\fpr=t \typ  F$ }
\DisplayProof\;
{\upshape (2)}
\AxiomC{$\Gamma, x\typ F\vDash_p G$}
\AxiomC{$\Gamma\vDash_p t\typ F$}
\AxiomC{$\Gamma\vDash_p u\typ G[t]$}
\TrinaryInfC{$\Gamma \vDash_p (t,u)\spr=u \typ  G[t]$ }
\DisplayProof
\end{lem}
\begin{proof}
Let $q\leqslant p$ and $q\forces \rho \typ \Gamma$.

\begin{enumerate}
\item We have $\vdash_q t\rho \typ F\rho$ and $\vdash_q u\rho \typ G[t]\rho$. By substitution we get $\vdash_q ((t,u)\fpr)\rho = t\rho \typ F\rho$. But $((t,u)\fpr)\rho \red_q t\rho$, thus $q\vdash ((t,u)\fpr)\rho \red_q t\rho \typ F\rho$. We have that $q\forces t\rho \typ F\rho$. Thus by Lemma~\ref{LEM:terms_onestepreductionimplyequiv} $q\forces ((t,u)\fpr)\rho \typ F\rho$ and $q\forces ((t,u)\fpr)\rho = t\rho \typ F\rho$.

\item Similarly we have $q\forces (t,u)\rho \spr \red u\rho \typ G[t\rho]$. Since $q\forces u\rho \typ G[t]\rho$, by Lemma~\ref{LEM:terms_onestepreductionimplyequiv}, we have that $q\forces ((t,u)\spr)\rho \typ G[t]\rho$ and $q\forces ((t,u)\spr)\rho = u\rho \typ G[t\rho]$. \qedhere
\end{enumerate}
\end{proof}

\begin{lem}\leavevmode\medskip
\label{soundprojtypeq}
\\[2mm]
{\upshape (1) }
\AxiomC{$\Gamma\vDash_p t\typ  \Sigma (x\typ F) G$ }
\UnaryInfC{$\Gamma \vDash_p t\fpr \typ  F$ }
\DisplayProof
\,
\AxiomC{$\Gamma\vDash_p t\typ  \Sigma (x\typ F) G$ }
\UnaryInfC{$\Gamma \vDash_p t\spr \typ  G[t\fpr]$ }
\DisplayProof
\,
{\upshape (2)}
\hspace{-2mm}
\AxiomC{$\Gamma \vDash_p t =u \typ  \Sigma (x\typ F) G$}
\UnaryInfC{$\Gamma \vDash_p t\fpr=u\fpr\typ F$}
\DisplayProof
\,
\AxiomC{$\Gamma \vDash_p t =u \typ  \Sigma (x\typ F) G$}
\UnaryInfC{$\Gamma \vDash_p t\spr=u\spr\typ G[t\fpr]$}
\DisplayProof 
\end{lem}
\begin{proof}

Let $q\leqslant p$ and $q\forces \rho \typ \Gamma$.
\begin{enumerate}
\item We have $q\forces t\rho \typ (\Sigma(x\typ F) G)\rho$. By the definition we have $q\forces (t\fpr)\rho \typ F\rho$ and $q\forces (t\spr)\rho \typ G[t\fpr]\rho$. Let $q\forces \rho = \sigma \typ \Gamma$. We have that $q\forces t\rho = t\sigma \typ (\Sigma(x\typ F) G)\rho$. By the definition $q\forces (t\fpr)\rho = (t\fpr)\sigma \typ F\rho$ and $q\forces (t\spr)\rho = (t\spr)\sigma \typ G[t\fpr]\rho$. 

\item We have $q\forces t\rho = u\rho \typ (\Sigma(x\typ F) G))\rho$. By the definition $q\forces (t\fpr)\rho=(u\fpr)\rho\typ F\rho$ and $q\forces (t\spr)\rho = (u\spr)\rho \typ G[t\fpr]\rho$.\qedhere
\end{enumerate}
\end{proof}

\begin{lem}
  \ \\[-4mm]
\label{soundpairext}
\begin{prooftree}
\AxiomC{$\Gamma\vDash_p t\typ  \Sigma(x\typ  F) G$}
\AxiomC{$\Gamma\vDash_p u\typ  \Sigma(x\typ  F) G$}
\AxiomC{$\Gamma\vDash_p t\fpr=u\fpr\typ  F$ }
\AxiomC{$\Gamma\vDash_p t\spr=u\spr\typ  G[t\fpr]$ }
\QuaternaryInfC{$\Gamma \vDash_p t=u \typ  \Sigma(x\typ F)G$ }
\end{prooftree}
\end{lem}
\begin{proof}
Let $q\leqslant p$ and $q\forces \rho \typ \Gamma$. We have $q\forces t\rho \typ (\Sigma(x\typ F) G)\rho$ and $q\forces u\rho \typ (\Sigma(x\typ F) G)\rho$. We also have $q\forces (t\fpr)\rho = (u\fpr) \rho \typ F\rho$ and $q\forces (t\spr) \rho = (u\spr) \rho \typ G[t\fpr]\rho$. By the definition $q\forces t\rho = u\rho \typ (\Sigma(x\typ F)G)\rho$.
\end{proof}

\begin{lem}
\label{soundnat}

{\upshape (1)} \AxiomC{$\Gamma\vdash_p$}
\UnaryInfC{$\Gamma\vDash_p \tnat$}
\DisplayProof
\,
{\upshape (2)} \AxiomC{$\Gamma\vdash_p$}
\UnaryInfC{$\Gamma \vDash_p 0 \typ  \tnat$}
\DisplayProof
\,
{\upshape (3)} \AxiomC{$\Gamma \vDash_p n \typ  \tnat$}
\UnaryInfC{$\Gamma \vDash_p \Suc\, n \typ  \tnat$}
\DisplayProof 
\,
{\upshape (4)} \AxiomC{$\Gamma \vDash_p n=m \typ \tnat$}
\UnaryInfC{$\Gamma \vDash_p \Suc\,n = \Suc\, m \typ \tnat$}
\DisplayProof
\end{lem}
\begin{proof}
(1) and (2) follow directly from the definition while (3) and (4) follow from Lemma~\ref{LEM:BaseTypeCanonicity} and locality.
\end{proof}

\begin{lem}
\label{soundnessnatrec}

\AxiomC{$\Gamma, x\typ \tnat \vDash_p F$}
\AxiomC{$\Gamma \vDash_p a_0 \typ  F[\zer]$}
\AxiomC{$\Gamma \vDash_p g\typ  \Pi (x\typ \tnat) (F[x]\rightarrow F[\Suc\,x])$}
\TrinaryInfC{$\Gamma \vDash_p \natrec(\lambda x. F)\,a_0\, g\typ  \Pi (x\typ \tnat) F$}
\DisplayProof
\end{lem}
\begin{proof}
Let $q\leqslant p$ and $q\forces \rho \typ \Gamma$. We have then that $q\vdash \rho \typ \Gamma$, hence, $\vdash_q (\natrec (\lambda x. F)\rho\,a_0\,g)\rho \typ (\Pi(x\typ \tnat) F)\rho$. Let $r\leqslant q$. Let $r\forces a \typ \tnat$, $r\forces b\typ \tnat$ and $r\forces a=b \typ \tnat$. By Lemma~\ref{LEM:BaseTypeCanonicity} there is a partition $r\covt S$ such that for each $s\in S$, $s\vdash a\rtred \xoverline{n}\typ \tnat$ and $s\vdash b\rtred \xoverline{n}\typ \tnat$. In order to show that $q\forces (\natrec (\lambda x. F)\,a_0\,g)\rho \typ (\Pi(x\typ \tnat) F)\rho$ we need to show that $r\forces (\natrec(\lambda x. F)\,a_0\, g)\rho\,a \typ F\rho[a]$, $r\forces (\natrec(\lambda x. F)\,a_0\, g)\rho\,b \typ F\rho[b]$, and $r\forces \natrec(\lambda x. F)\,a_0\, g)\rho\,a =(\natrec(\lambda x. F)\,a_0\, g)\rho\,b \typ F\rho[a]$. By locality it will be sufficient to show that for each $s\in S$ we have $s\forces (\natrec(\lambda x. F)\,a_0\, g)\rho\,a \typ F\rho[a]$, $s\forces (\natrec(\lambda x. F)\,a_0\, g)\rho\,b \typ F\rho[b]$, and $s\forces (\natrec(\lambda x. F)\,a_0\, g)\rho\,a= (\natrec(\lambda x. F)\,a_0\, g)\rho\,b\rho\typ F\rho[a]$.  We have that 
\begin{align*}
& s\vdash (\natrec(\lambda x. F)\,a_0\, g)\rho\,a  \rtred (\natrec(\lambda x. F)\,a_0\, g)\rho\,\xoverline{n} \typ F\rho[a]\\
& s\vdash (\natrec(\lambda x. F)\,a_0\, g)\rho\,b  \rtred (\natrec(\lambda x. F)\,a_0\, g)\rho\,\xoverline{n} \typ F\rho[b]
\end{align*}
Let $\xoverline{n} \coloneqq \Suc^{k}\,0$. By induction on $k$. If $k=0$ then 
\begin{align*}
& s\vdash (\natrec(\lambda x. F)\,a_0\, g)\rho\,a  \rtred (\natrec(\lambda x. F)\,a_0\, g)\rho\,\zer\red a_0\rho \typ F\rho[a]\\
& s\vdash (\natrec(\lambda x. F)\,a_0\, g)\rho\,b  \rtred (\natrec(\lambda x. F)\,a_0\, g)\rho\,\zer\red a_0\rho \typ F\rho[b]
\end{align*}
By Lemma~\ref{LEM:equivalentTypesHaveEqualPers} we have then that 
\begin{align*}
s\forces (\natrec(\lambda x. F)\,a_0\, g)\rho\,a  = a_0\rho \typ F\rho[a] &\hphantom{space}
s\forces (\natrec(\lambda x. F)\,a_0\, g)\rho\,b  = a_0\rho \typ F\rho[b]
\end{align*}

Since $s\forces a=b\typ \tnat$ we have $s\forces (\rho,a)=(\rho,b)\typ (\Gamma,x\typ \tnat)$ and thus $s\forces F\rho[a]=F\rho[b]$. By Lemma~\ref{LEM:equivalentTypesHaveEqualPers}, symmetry and transitivity 
$s\forces (\natrec(\lambda x. F)\,a_0\, g)\rho\,a= (\natrec(\lambda x. F)\,a_0\, g)\rho\,b\typ F\rho[a]$.

Assume the statement holds for $k\leq \ell$. Let $\xoverline{n} = \Suc\,\xoverline{\ell}$. We have then
\begin{align*}
& s\vdash (\natrec(\lambda x. F)\,a_0\, g)\rho\,a  \rtred  g\rho \,\xoverline{\ell}\, ((\natrec(\lambda x. F)\,a_0\, g)\rho\,\xoverline{\ell}) \typ F\rho[\Suc\,\xoverline{\ell}]\\
& s\vdash (\natrec(\lambda x. F)\,a_0\, g)\rho\,b  \rtred  g\rho \,\xoverline{\ell}\, ((\natrec(\lambda x. F)\,a_0\, g)\rho\,\xoverline{\ell}) \typ F\rho[\Suc\,\xoverline{\ell}]
\end{align*}
By IH $s\forces ((\natrec(\lambda x. F)\,a_0\, g)\rho\,\xoverline{\ell}) \typ F\rho[\xoverline{\ell}]$. But we have $\Gamma \vDash_p g\typ  \Pi (x\typ \tnat) (F[x]\rightarrow F[\Suc\,x])$ and thus
$s\forces  (g\rho)\,\xoverline{\ell}\,((\natrec(\lambda x. F)\,a_0\, g)\rho\,\xoverline{\ell})\typ F\rho[\Suc\,\xoverline{\ell}]$. By Corollary~\ref{COR:terms_reductionimplyequiv}, symmetry and transitivity we get that $s\forces (\natrec(\lambda x. F)\,a_0\, g)\rho\,a \typ F\rho[\Suc\,\xoverline{\ell}]$, 
$s\forces (\natrec(\lambda x. F)\,a_0\, g)\rho\,b \typ F\rho[\Suc\,\xoverline{\ell}]$, and
$s\forces (\natrec(\lambda x. F)\,a_0\, g)\rho\,a = (\natrec(\lambda x. F)\,a_0\, g)\rho\,b \typ F\rho[\Suc\,\xoverline{\ell}]$. But $s\forces a = \Suc\,\xoverline{\ell} \typ \tnat$, thus, $s\forces F\rho[a]=F\rho[\Suc\,\xoverline{\ell}]$. By Lemma~\ref{LEM:equivalentTypesHaveEqualPers} we get then that
$s\forces (\natrec(\lambda x. F)\,a_0\, g)\rho\,a = (\natrec(\lambda x. F)\,a_0\, g)\rho\,b \typ F\rho[a]$.

As indicated above, this is sufficient to show $q\forces (\natrec(\lambda x. F)\,a_0\, g)\rho \typ  \Pi (x\typ \tnat) F\rho$. 

Given $q\forces \rho=\sigma\typ \Gamma$. Similarly we can show $q\forces (\natrec(\lambda x. F)\,a_0\, g)\sigma \typ  \Pi (x\typ \tnat) F\sigma$. 

To show that $q\forces (\natrec(\lambda x. F)\,a_0\, g)\rho=(\natrec(\lambda x. F)\,a_0\, g)\sigma \typ  \Pi (x\typ \tnat) F\rho$ we need to show that whenever $r\forces a\typ F\rho$ for some $r\leqslant q$ we have $r\forces  (\natrec(\lambda x. F)\,a_0\, g)\rho\,a= (\natrec(\lambda x. F)\,a_0\, g)\sigma\,a\typ F\rho[a]$. Let $r\forces a\typ F$ for $r\leqslant q$. By Lemma~\ref{LEM:BaseTypeCanonicity} we have a partition $r\covt S$ where for each $s\in S$ we have $s\vdash a\rtred \xoverline{n}\typ \tnat$. As  above it is  sufficient to show $s\forces  (\natrec(\lambda x. F)\,a_0\, g)\rho\,a= (\natrec(\lambda x. F)\,a_0\, g)\sigma\,a\typ F\rho[a]$ for all $s\in S$.

Let $\xoverline{n} \coloneqq \Suc^{k}\,\zer$. By induction on $k$. If $k=0$ then as above
\begin{align*}
s\forces (\natrec(\lambda x. F)\,a_0\, g)\rho\,a  = a_0\rho \typ F\rho[a] &\hphantom{space}
s\forces (\natrec(\lambda x. F)\,a_0\, g)\sigma\,a  = a_0\sigma \typ F\sigma[b]
\end{align*}

Since $r\forces \rho =\sigma\typ \Gamma$ we have $s\forces (\rho,a)=(\sigma,a)\typ (\Gamma,x\typ \tnat)$. We have then that $s\forces F\rho[a]=F\sigma[a]$. But we also have that $s\forces a_0\rho =a_0\sigma\typ F\rho[a]$. By Lemma~\ref{LEM:equivalentTypesHaveEqualPers}, symmetry and transitivity it then follows that $s\forces (\natrec(\lambda x. F)\,a_0\, g)\rho\,a=(\natrec(\lambda x. F)\,a_0\, g)\sigma\,a\typ F\rho[a]$.

Assume the statement holds for $k\leq \ell$. Let $\xoverline{n} = \Suc\,\xoverline{\ell}$. As before we have that
\begin{align*}
& s\vdash (\natrec(\lambda x. F)\,a_0\, g)\rho\,a  \rtred  g\rho \,\xoverline{\ell}\, ((\natrec(\lambda x. F)\,a_0\, g)\rho\,\xoverline{\ell}) \typ F\rho[\Suc\,\xoverline{\ell}]\\
& s\vdash (\natrec(\lambda x. F)\,a_0\, g)\sigma\,a  \rtred  g\sigma \,\xoverline{\ell}\, ((\natrec(\lambda x. F)\,a_0\, g)\sigma\,\xoverline{\ell}) \typ F\sigma[\Suc\,\xoverline{\ell}]
\end{align*}
By IH 
$s\forces (\natrec(\lambda x. F)\,a_0\, g)\rho\,\xoverline{\ell} =(\natrec(\lambda x. F)\,a_0\, g)\sigma\,\xoverline{\ell} \typ F\rho[\xoverline{\ell}]$.
But $s\forces g\rho = g\sigma \typ (\Pi (x\typ \tnat) (F[x]\rightarrow F[\Suc\,x]))\rho$, thus
 $s \forces g\rho\,\xoverline{\ell}\,((\natrec(\lambda x. F)\,a_0\, g)\rho\,\xoverline{\ell}) = g\sigma\,\xoverline{\ell}\,((\natrec(\lambda x. F)\,a_0\, g)\sigma\,\xoverline{\ell}) \typ F\rho[\Suc\,\xoverline{\ell}]$

But $s\forces F\rho[\Suc\,\xoverline{\ell}]=F\sigma[\Suc\,\xoverline{\ell}]$ and $r_i\forces F\rho[a]=F\rho[\Suc\,\xoverline{\ell}]$. By Lemma~\ref{LEM:equivalentTypesHaveEqualPers}, symmetry and transitivity we have then that 
$s\forces (\natrec(\lambda x. F)\,a_0\, g)\rho\,a  = (\natrec(\lambda x. F)\,a_0\, g)\sigma\,a \typ F\rho[a]$

Which is sufficient to show 
$q\forces (\natrec(\lambda x. F)\,a_0\, g)\rho= (\natrec(\lambda x. F)\,a_0\, g)\sigma \typ \Pi(x\typ \tnat) F\rho$
\end{proof}

\begin{lem}
\label{soundnatrec0}

\AxiomC{$\Gamma, x\typ \tnat \vdash F$}
\AxiomC{$\Gamma \vDash_p a_0 \typ  F[\zer]$}
\AxiomC{$\Gamma \vDash_p g \typ  \Pi (x\typ \tnat) (F[x] \rightarrow F[\Suc\,x])$}
\TrinaryInfC{$\Gamma \vDash_p \natrec (\lambda x. F)\,a_0\, g \,\zer = a_0 \typ  F[\zer]$}
\DisplayProof
\end{lem}
\begin{proof}
Let $q\leqslant p$ and $q\forces \rho \typ \Gamma$. We have $\vdash_q \rho \typ \Gamma$ and thus we get that $\vdash_q (\natrec(\lambda x. F)\,a_0\,g\,\zer)\rho = a_0\rho \typ F\rho[\zer]$. But $(\natrec(\lambda x. F)\,a_0\,g\,\zer)\rho \red a_0\rho$. Thus $q\vdash (\natrec(\lambda x. F)\,a_0\,g\,\zer)\rho \red a_0\rho\typ F\rho[\zer]$. But $q\forces a_0\rho \typ F\rho[\zer]$. By Lemma~\ref{LEM:terms_onestepreductionimplyequiv} we have $q\forces (\natrec(\lambda x. F)\,a_0\,g\,\zer)\rho =a_0\rho\typ F\rho[\zer]$.\smallskip
\end{proof}

\begin{lem}
\label{soundnatrecsuc}
\AxiomC{$\Gamma, x\typ \tnat \vDash_p F$}
\AxiomC{$\Gamma \vDash_p a_0 \typ  F[\zer]$}
\AxiomC{$\Gamma \vDash_p n \typ  \tnat$}
\AxiomC{$\Gamma \vDash_p g \typ  \Pi (x\typ \tnat) (F[x] \rightarrow F[\Suc \,x])$}
\QuaternaryInfC{$\Gamma \vDash_p \natrec (\lambda x. F)\,a_0\, g \,(\Suc\, n) = g\,n \,(\natrec (\lambda x. F)\,a_0\,g\,n) \typ  F[\Suc\,n]$}
\DisplayProof

\end{lem}
\begin{proof}
Let $q\leqslant p$ and $q\forces \rho \typ \Gamma$. We have $q\forces n\rho \typ \tnat$. By Lemma~\ref{LEM:BaseTypeCanonicity} there is a partition $q\covt S$ such that for each $s\in S$ there is $m\in \nats$ and $s\vdash n\rho \rtred \xoverline{m}\typ \tnat$. Thus  $s\vdash \Suc\,n\rho \rtred \Suc\,\xoverline{m}\typ \tnat$. We have then that 
\begin{align*}
s\vdash &(\natrec (\lambda x. F)\,a_0\, g \,(\Suc\, n))\rho  \rtred  (\natrec (\lambda x. F)\,a_0\, g)\rho\,(\Suc\, \xoverline{m})\\
&\rtred g\rho\,\xoverline{m}\,((\natrec\,(\lambda x. F)\,a_0\,g)\rho\,\xoverline{m})\typ F\rho[\Suc\,\xoverline{m}]
\end{align*}

But $s \vdash (g\,n \,(\natrec (\lambda x. F)\,a_0\,g\,n))\rho \rtred g\rho\,\xoverline{m}\,((\natrec\,(\lambda x. F)\,a_0\,g)\rho\,\xoverline{m})\typ F\rho[\Suc\,\xoverline{m}]$. By Corollary~\ref{COR:terms_reductionimplyequiv}, symmetry and transitivity 
$s \forces (\natrec (\lambda x. F)\,a_0\, g \,(\Suc\, n))\rho =(g\,n \,(\natrec (\lambda x. F)\,a_0\,g\,n))\rho \typ F\rho[\Suc\,\xoverline{m}]$.

Since $s\forces n\rho =\xoverline{m} \typ \tnat$ we have that $s\forces F\rho[\Suc\,\xoverline{m}] = F\rho[\Suc\, n\rho]$.
By Lemma~\ref{LEM:equivalentTypesHaveEqualPers} we thus have that
$ s \forces (\natrec (\lambda x. F)\,a_0\, g \,(\Suc\, n))\rho =(g\,n \,(\natrec (\lambda x. F)\,a_0\,g\,n))\rho \typ F[\Suc\,n]\rho$.

By locality  
$q \forces (\natrec (\lambda x. F)\,a_0\, g \,(\Suc\, n))\rho =(g\,n \,(\natrec (\lambda x. F)\,a_0\,g\,n))\rho \typ F[\Suc\,n]\rho$   
\end{proof}

\begin{lem}
\label{soundnatreccong}
\AxiomC{$\Gamma, x\typ \tnat\vDash_p F = G$}
\AxiomC{$\Gamma \vDash_p a_0=b_0 \typ  F[\zer]$}
\AxiomC{$\Gamma \vDash_p g=h \typ  \Pi (x\typ \tnat) (F[x] \rightarrow F[\Suc \,x])$}
\TrinaryInfC{$\Gamma \vDash_p \natrec (\lambda x. F)\,a_0\, g = \natrec (\lambda x. G)\,b_0\, h\typ  \Pi(x\typ \tnat) F$}
\DisplayProof\medskip
\end{lem}
\begin{proof}
The proof follows by an argument similar to that used to prove Lemma~\ref{soundnessnatrec}.
\end{proof}

For the congruence rules, soundness follows from Theorem ~\ref{properties}. Soundness for rules of $\tempt$, $\tunit$ and $\tbool$ follow similarly to those of $\tnat$. Soundness for the rules of $\tuniv$ follows similarly to soundness of typing rules. We have then the following corollary:
\begin{cor}[Soundness]
\label{LEM:soundness}
If $\Gamma \vdash_p J$ then $\Gamma\vDash_p J$
\end{cor}

\begin{thm}[Fundamental Theorem]
\label{THM:fundamentalThm}
If $\vdash_p J$ then $p\forces J$.
\end{thm}
\begin{proof}
Follows from Corollary~\ref{LEM:soundness}.
\end{proof}
\section{Markov's principle}
\label{SEC:markov}
Now we have enough machinery to show the independence of $\MP$ from type theory. The idea is that if a judgment $J$ is derivable in type theory (i.e.\ $\vdash J$) then it is derivable in the forcing extension (i.e.\ $\vdash_{\emptyset} J$) and by Theorem~\ref{THM:fundamentalThm} it holds in the interpretation (i.e.\ $\forces J$). It thus suffices to show that there no $t$ such that $\forces t\typ \tMP$ to establish the independence of $\MP$ from type theory. First we recall the formulation of $\MP$.
\begin{align*}
\tMP \coloneqq\Pi (h\typ \tnat\rightarrow \tbool)[\neg\neg (\Sigma (x\typ \tnat)\,\iszero\,(h\,x)) \rightarrow \Sigma (x\typ \tnat)\,\iszero\,(h\,x)]
\end{align*}
where $\iszero \typ \tbool \rightarrow \tuniv$ is given by $\iszero \coloneqq \lambda y.\boolrec (\lambda x. \tuniv)\, \tunit\,\tempt\,y$.
 

\begin{lem}
\label{LEM:markovConclusionNoWitness}
There is no term $t$ such that $\forces t\typ \Sigma(x\typ \tnat)\,\iszero\,(\f\,x)$.
\end{lem}
\begin{proof}
Assume $\forces t\typ \Sigma(x\typ \tnat)\,\iszero\,(\f\,x)$ for some $t$. We then have $\forces t\fpr \typ \tnat$ and $\forces t\spr\typ \iszero\,(\f\,t\fpr)$. By Lemma~\ref{LEM:BaseTypeCanonicity} one has a partition $\emptyset \covt S$ where for each $q\in S$, $q\vdash t\fpr \rtred \xoverline{m}\typ \tnat$. Hence $q\vdash \iszero\,(\f\,t\fpr) \rtred \iszero\,(\f\,\xoverline{m})$ and by Lemma~\ref{LEM:types_reductionimplyequiv} $q \forces \iszero\,(\f\,t\fpr)= \iszero\,(\f\,\xoverline{m})$. But, by definition, the partition $S$ must contain a condition, say $r$, such that $r(k) = 1$ whenever $k \in \dom{r}$ (this holds vacuously for $\emptyset\covt \{\emptyset\}$). Let $r\vdash t\fpr \rtred \xoverline{n}\typ \tnat$. Assume $n \in \dom{r}$, then $r\vdash \iszero\,(\f\,t\fpr) \rtred \iszero\,(\f\,\xoverline{n}) \rtred \tempt$. By monotonicity, from $\forces t\spr \typ \iszero\,(\f\,t\fpr)$ we get $r\forces t\spr \typ \iszero\,(\f\,t\fpr)$. But $r \vdash \iszero\,(\f\,t\fpr) \rtred \tempt$ thus $r \forces  \iszero\,(\f\,t\fpr) = \tempt$. Hence, by Lemma~\ref{LEM:equivalentTypesHaveEqualPers}, $r\forces t\spr \typ \tempt$ which is impossible, thus contradicting our assumption. If on the other hand $n \notin \dom{r}$ then, since $r\covt\{ r(n\mapsto \zer),r(n\mapsto \one)\}$, we can apply the above argument with $r(n\mapsto \one)$ instead of $r$.
\end{proof}

\begin{lem}
\label{LEM:markovNotForcible}
There is no term $t$ such that $\forces t\typ \tMP$.
\end{lem}
\begin{proof}
Assume $\forces t\typ \tMP$ for some $t$. From the definition, whenever $\forces g \typ \tnat\rightarrow \tbool$ we have $\forces t\,g \typ \neg \neg (\Sigma(x\typ \tnat)\,\iszero\,(g\,x))\rightarrow \Sigma (x\typ \tnat)\,\iszero\,(g\,x)$. Since by Corollary~\ref{COR:genericwelltyped}, $\forces \f \typ \tnat\rightarrow \tbool$ we have $\forces t\,\f \typ \neg \neg (\Sigma(x\typ \tnat)\,\iszero\,(\f\,x))\rightarrow \Sigma(x\typ \tnat)\,\iszero\,(\f\,x)$. Since by Lemma~\ref{LEM:markovpremisewitnesswelltyped} $\forces \w\typ \neg\neg( \Sigma(x\typ \tnat)\,\iszero\,(\f\,x))$ we have $\forces (t\,\f)\,\w \typ \Sigma (x\typ \tnat)\,\iszero\,(\f\,x)$ which is impossible by Lemma~\ref{LEM:markovConclusionNoWitness}.
\end{proof}

From Theorem~\ref{THM:fundamentalThm}, Lemma~\ref{LEM:markovNotForcible}, and Lemma~\ref{LEM:extension} we can then conclude:
\mainresult*

\subsection{Many Cohen reals}

We extend the type system in Section~\ref{SEC:Forcingextensionoftypetheory} further by adding a generic point $\f_q$ for each condition $q$. The introduction and conversion rules for $\f_q$ are given by:\smallskip

\begin{center}
\AxiomC{$\Gamma\vdash_p$}
\UnaryInfC{$\Gamma\vdash_p \f_q \typ \tnat\rightarrow \tbool$}
\DisplayProof
\,
\AxiomC{$\Gamma\vdash_p$}
\RightLabel{$n\in\dom{q}$}
\UnaryInfC{$\Gamma\vdash_p \f_q\,\xoverline{n} = 1$}
\DisplayProof\,
\AxiomC{$\Gamma\vdash_p$}
\RightLabel{$n\notin\dom{q}, n\in \dom{p}$}
\UnaryInfC{$\Gamma\vdash_p \f_q\,\xoverline{n} = p(n)$}
\DisplayProof\smallskip
\end{center}

\noindent
With the reduction rules
\begin{center}
\AxiomC{$n\in \dom{q}$}
\UnaryInfC{$\f_q\,\xoverline{n} \rightarrow 1$}
\DisplayProof
\,
\AxiomC{$n\notin \dom{q}, n\in \dom{p}$}
\UnaryInfC{$\f_q\,\xoverline{n} \rightarrow_p p(n)$}
\DisplayProof\medskip
\end{center}
We observe that with these added rules the reduction relation is still monotone.
For each $\f_q$ we add a term:
\begin{center}
\AxiomC{$\Gamma \vdash_p$}
\UnaryInfC{$\Gamma \vdash_p \w_q\typ  \neg \neg (\Sigma(x\typ \tnat)\,\iszero\,(\f_q\,x)) $}
\DisplayProof\smallskip
\end{center}
Finally we add a term $\mw$ witnessing the negation of $\MP$\smallskip
\begin{center}
\AxiomC{$\Gamma \vdash_p$}
\UnaryInfC{$\Gamma \vdash_p \mw\typ \neg \tMP$}
\DisplayProof 
\smallskip
\end{center}
By analogy to Corollary~\ref{COR:genericwelltyped} we have 
\begin{lem}
\label{LEM:qGenericPointForced}
$\forces \f_q \typ \tnat\rightarrow \tbool$ for all $q$.
\end{lem}

\begin{lem}
\label{LEM:qMarkovPermiseWitnessForced}
$\forces \w_q \typ  \neg \neg (\Sigma(x\typ \tnat)\,\iszero\,(\f_q\,x))$ for all $q$.
\end{lem}
\begin{proof}
Assume $p\forces t\typ \neg (\Sigma(x\typ \tnat) \iszero\,(\f_q\,x))$ for some $p$ and $t$. Let $m \notin \dom{q} \cup \dom{p}$, we have $p(m\mapsto \zer)\forces \f_q\,\xoverline{m}=\zer$. Thus $p(m \mapsto \zer)\forces (\xoverline{m},\zer) \typ  \Sigma(x\typ \tnat)\,\iszero\,(\f_q\,x)$ and $p(m\mapsto \zer) \forces t\,(\xoverline{m},\zer) \typ \tempt$ which is impossible.
\end{proof}

\begin{lem}
\label{LEM:qMarkovConclusionNoWitness}
There is no term $t$ for which $q\forces t\typ \Sigma(x\typ \tnat)\,\iszero\,(\f_q\,x)$.
\end{lem}
\begin{proof}
Assume $q\forces t\typ \Sigma(x\typ \tnat)\,\iszero\,(\f_q\,x)$ for some $t$. We then have $q\forces t\fpr \typ \tnat$ and $q\forces t\spr\typ \iszero\,(\f_q\,t\fpr)$. By Lemma~\ref{LEM:BaseTypeCanonicity} one has a partition $q \covt \{q_1,\dots,q_n\}$ where for each $i$, $t\fpr \rtred_{q_i} \xoverline{m}_i$ for some $\xoverline{m}_i \in \nats$. Hence $q_i\vdash \iszero\,(\f_q\,t\fpr) \rtred \iszero\,(\f_q\,\xoverline{m}_i)$ . But any partition of $q$ contain a condition, say $q_j$, where $q_j(k) = 1$ whenever $k \notin \dom{q}$ and $k\in \dom{q_j}$. Assume $m_j \in \dom{q_j}$. If $m_j \in\dom{q}$ then $q_j\vdash \f_q\,m_j \red \one\typ \tbool$ and if $m_j\notin\dom{q}$ then $q_j\vdash \f_q\,\xoverline{m}_j\red q_j(k)\coloneqq \one\typ \tbool$. Thus $q_j\vdash \iszero\,(\f_q\,t\fpr) \rtred \tempt$ and by Lemma~\ref{LEM:types_reductionimplyequiv} $q_j \forces \iszero\,(\f\,t\fpr) =\tempt$. From $\forces t\spr \typ \iszero\,(\f\,t\fpr)$ by monotonicity and Lemma~\ref{LEM:equivalentTypesHaveEqualPers} we have $q_j\forces t\spr \typ \tempt$ which is impossible. If on the other hand $m_j \notin \dom{q_j}$ then since $q_j\covt \{q_j(m_j\mapsto \zer),q_j(m_j\mapsto \one)\}$ we can apply the above argument with $q_j(m_j\mapsto \one)$ instead of $q_j$.
\end{proof}
\begin{lem}
$\forces \mw\typ \neg \tMP$
\end{lem}
\begin{proof}
Assume $p \forces t\typ \tMP$ for some $p$ and $t$. Thus whenever $q\leqslant p$ and $q\forces u \typ \tnat\rightarrow \tbool$ then $q\forces t\,u \typ \neg \neg (\Sigma(x\typ \tnat)\,\iszero\,(u\,x))\rightarrow (\Sigma(x\typ \tnat)\,\iszero\,(u\,x))$. But we have $q\forces \f_q \typ \tnat\rightarrow \tbool$ by Lemma~\ref{LEM:qGenericPointForced}. Hence $q\forces t\,\f_q \typ \neg \neg (\Sigma(x\typ \tnat) \iszero\,(\f_q\,x))\rightarrow (\Sigma(x\typ \tnat) \iszero\,(\f_q\,x))$. But $q\forces \w_q \typ \neg \neg (\Sigma(x\typ \tnat) \iszero\,(\f_q\,x))$ by Lemma~\ref{LEM:qMarkovPermiseWitnessForced}. Thus $q\forces (t\,\f_q)\,\w_q \typ \Sigma(x\typ \tnat)\,\iszero\,(\f_q\,x)$ which is impossible by Lemma~\ref{LEM:qMarkovConclusionNoWitness}.
\end{proof}

We have then that this extension is sound with respect to the interpretation. Hence we have shown the following statement.

\begin{thm}
There is a consistent extension of $\MLTT$ where $\neg \tMP$ is derivable.
\end{thm}

Recall that $\mathsf{dne}\coloneqq \Pi(A\typ \tuniv)(\neg\neg A\rightarrow A)$. We have then a term $\vdash t\typ \mathsf{dne} \rightarrow \tMP$. Thus in this extension we have a term $\vdash_\emptyset \lambda x. \mw\,(t\, x)\typ \neg \mathsf{dne}$.

\section*{Acknowledgement}
We would like to thank Simon Huber, Thomas Streicher, Mart\'{\i}n Escard\'{o} and Chuangjie Xu.
\bibliographystyle{alpha}
\bibliography{chapMarkovIndependence.bbl}

\end{document}